\documentclass[a4paper,12pt]{article}
\usepackage{amsmath}
\usepackage{amsthm}
\usepackage{amsfonts}
\usepackage{amssymb}
\usepackage{mathrsfs}
\usepackage{graphicx}
\usepackage{color}
\usepackage[colorlinks=true, allcolors=blue]{hyperref}

\newtheorem{theorem}{Theorem}
\newtheorem{lemma}{Lemma}
\newtheorem{remark}{Remark}

\usepackage[english]{babel}

\numberwithin{equation}{section}

\hypersetup{ pdftitle={Draft},
	pdfauthor={Mirda Prisma Wijayanto, Fiki Taufik Akbar, Bobby Eka Gunara}, 
	pdfkeywords={Global Esistence, Einstein-Maxwell-Higgs}, 
	bookmarksnumbered, pdfstartview={FitH}, urlcolor=blue, }

\begin{document}
	\pagestyle{plain}
	
	
	
	
	\title{\LARGE\textbf{Classical Solutions of Higher Dimensional Einstein-Maxwell-Higgs System With Nontrivial Potential: Global Existence and Completeness}}

		\author{Mirda Prisma Wijayanto$^1$\footnote{Corresponding author}, Fiki Taufik Akbar$^2$, Bobby Eka Gunara$^2$ \\ 
		\\
		$^1$ \textit{\small Department of Physics, Faculty of Mathematics and Natural Sciences,}\\
		\textit{\small Universitas Jenderal Soedirman}\\
		\textit{\small Jl. Dr. Soeparno no. 61 Purwokerto, Indonesia, 53122}\\
		$^2$\textit{\small Theoretical High Energy Physics Research Division, }\\
		\textit{\small Faculty of Mathematics and Natural Sciences, Institut Teknologi Bandung}\\
		\textit{\small Jl. Ganesha no. 10 Bandung, Indonesia, 40132}
		\\
		\\
		\small email: mirda.foundation@gmail.com, ftakbar@itb.ac.id, bobby@itb.ac.id}
		
	\date{\today}
	\maketitle
	
	
	
	
	\begin{abstract}
		We study the Cauchy problem of higher dimensional Einstein-Maxwell-Higgs system in the framework of Bondi coordinates. As a first step, the problem is reduced to a single first-order integro-differential equation by defining a generalized ansatz function. Then, we employ contraction mapping to show that there exists the unique fixed point of the problem. For a given small initial data, we prove the existence of a global classical solution. Finally, by introducing local mass and local charge functions in higher dimensions, we also show the completeness property of the spacetimes. 
	\end{abstract}
	
	
	%
	%
	%
	%
	%
	\tableofcontents
	{\allowdisplaybreaks
	\section{Introduction}\label{Sec1}
	\subsection{Motivation}	
	One of the interesting problems in general relativity is that black holes are characterized by only three parameters consisting of mass, angular momentum, and electromagnetic charge which are related to Gauss' Law, such as Kerr-Newman black holes \cite{Ruffini,Kerr,Newman}. It follows that a black hole cannot support an external scalar field because Gauss law does not apply to scalars. This result has been proven by \cite{Sudarsky} for a scalar field that is minimally coupled to gravity and by \cite{Saa,Saa2,Mayo} for a scalar field that is not minimally coupled to asymptotically flat space. In its development, it has been shown that scalar hairy black holes exist if we apply certain scalar potentials that violate the weak energy condition \cite{Nucamendi}. One of the attempts to study hairy black holes is to create a mathematical model that leads to a theory to describe the interaction between the matter field described by gauge theory with spacetime according to Einstein's equations. We then refer to this model as the  the Einstein-Maxwell-Higgs (EMH) system.
	
	The Cauchy problem for the Einstein-Maxwell-Higgs system has been studied over the past few years. A regular localized solution of the classical theory of the gravitational field coupled to the electromagnetic field and to an auxiliary scalar is presented in \cite{Clement}. Motivated by Christodoulou (\cite{Chris1}-\cite{Chris4}), Chae \cite{Chae1} studied the global unique existence of classical solutions and the decay estimates of the solutions to the four dimensional Einstein-Maxwell-Higgs system for small initial data under the spherical symmetry.
	
	Over the past few decades, there has been renewed interest in the possibility that spacetime has higher dimensions ($D\geq 4$). This idea raises challenging new issues and also has profound implications for the generalization of Einstein's general theory of relativity that may give the different properties from four-dimensional theories (see \cite{Hollands}-\cite{Coley}). In its development, higher dimensional theory is widely applied as the basis for the development of several theories such as Klauza-Klein theory, string theory, superstring theory, and so on to study the concept of unification of forces. From this point of view, we interest to extend the study of Cauchy problem for the Einstein-Maxwell-Higgs system in higher dimensions. In general relativity, the aim of this study is to provide an answer to the question of what the initial conditions are for the Einstein equation coupled with the Maxwell-Higgs equation, which gives rise to a geodesically complete spacetime that asymptotically approaches the Minkowski spacetime in the infinite future.
	
	\subsection{EMH System in Higher Dimensions}
	In our system, we use the Bondi coordinates which can be described as follows. Let $r$ denotes the radial coordinate, where $r=0$ describes the central world line and the world lines $r=r_0$ in each half-plane are all timelike. With respect to this radial coordinate, we can define the functions $F$ and $G$, which is $F,G\rightarrow 0$ as $r\rightarrow\infty$ satisfying the asymptotic Ricci flatness condition. We also denote $u$ as the timelike coordinate which is constant on the future light cone of each point on the central world line. Furthermore, $u=0$ describes the initial future light cone, and $u$ tends to the proper time as $r_0\rightarrow \infty$ on a world line $r=r_o$. We only consider the future of some initial future light cone with vertex at the center. Hence, the range of the coordinates $u$ and $r$ are $0\leq u<\infty$ and $0\leq r<\infty$ respectively.
	
	Using the coordinate system that has been described above, we write the spacetime metric in the form
	\begin{align}\label{metric}
	\mathrm{d}s^2 = - e^{2F(u,r)}\mathrm{d}u^2 - 2e^{F(u,r)+G(u,r)}\mathrm{d}u\mathrm{d}r + r^2 \mathrm{d}\Omega^{D-2},
	\end{align}
	where the spatial metric $\mathrm{d}\Omega^{D-2}$ describes ($D-2$)-spatially compact manifold. If the $(D-2)$-spatial geometry is sphere, then the metric (\ref{metric}) becomes flat as $r\rightarrow\infty$. The motivation in choosing such a coordinate system is based on the Cosmic Censorship Conjecture which states that there exists a solution of the Einstein equations for a given arbitrary asymptotically flat initial data which is a globally hyperbolic spacetime possessing a complete future null infinity \cite{Chris1}.
			
	The action of our system has the form
	\begin{align}\label{action}
	S=\int d^Dx \sqrt{-\det({g_{\mu\nu}})}\left\{-\frac{1}{4}\tilde{F}_{\mu\alpha}\tilde{F}_{\nu\beta}g^{\mu\nu}g^{\alpha\beta}+g^{\mu\nu}D_\mu\phi (D_\nu\phi)^* - V(|\phi|)\right\},
	\end{align}
	where $\det({g_{\mu\nu}})$ is the determinant of the metric (\ref{metric}), $\phi=\phi_1+i\phi_2$ is a complex scalar field,  $D_\mu\phi:=\partial_\mu\phi+iA_\mu\phi$ is the (gauge) covariant derivative, and $\tilde{F}_{\mu\nu}:=\partial_\mu A_\nu - \partial_\nu A_\mu$ is the electromagnetic tensor satisfies 
	\begin{align}\label{nabla F}
	\nabla_\nu \tilde{F}^{\mu\nu}= J^\mu,
	\end{align}
	and
	\begin{align}
	J^\mu = \mathrm{Im}\left\{\phi (D^\mu\phi)^*\right\}.
	\end{align}
	Also, we assume that $A_i(u,r)=0$, with $i\geq 1$ and the scalar potential $V(|\phi|)$ satisfies
	\begin{align}\label{potential}
	|V(|\phi|)|+\left|\frac{\partial V(|\phi|)}{\partial \phi^*}\right||\phi|+\left|\frac{\partial^2V(|\phi|)}{\partial\phi\partial\phi^*}\right||\phi|^2\leq K_0 |\phi|^{p+1},~~~\forall\phi\in\mathbb{C},
	\end{align}
	where $K_0$ is a positive constant, and $p\in \mathbb{R}^+$.
	Now, we define a $D-$dimensional local charge function
	\begin{align}\label{local charge}
	Q(u,r):=\int_{B(0,r)} J^0 dv = Vol(\Sigma^{D-2}) \int_{0}^{r}J^0 \sqrt{-\det(g_{\mu\nu})}s^{D-2}~\mathrm{d}s,
	\end{align}
	where $Vol(\Sigma^{D-2})$ represents the volume of $(D-2)$ compact manifold. Taking integration of (\ref{nabla F}) and using the definition of (\ref{local charge}), we obtain
	\begin{align}\label{A0}
	A_0=\frac{1}{Vol(\Sigma^{D-2})} \int_{0}^{r} Qe^{F+G}s^{-(D-2)}~\mathrm{d}s.
	\end{align}
			
	From the action (\ref{action}), we obtain the equation of motion
	\begin{align}\label{EMH}
	&\frac{\partial}{\partial u}\left[r \frac{\partial\phi}{\partial r}+\frac{(D-2)}{2}\phi \right]-\frac{e^{F-G}}{2}\frac{\partial}{\partial r}\left[r\frac{\partial\phi}{\partial r}+\frac{(D-2)}{2}\phi\right]\nonumber\\
	&=\frac{1}{2}\left(\frac{\partial\phi}{\partial r}\right)\left[\frac{\hat{R}(\hat{\sigma})}{(D-2)}e^{F+G}-\frac{(D-2)}{2}e^{F-G}\right]\nonumber\\
	&-\frac{Q^2}{(D-2)r^2}e^{F+G}\left(\frac{\partial\phi}{\partial r}\right)
	-\frac{8\pi r^2}{(D-2)} e^{F+G}\left(\frac{\partial\phi}{\partial r}\right)V(|\phi|)\nonumber\\
	&-i\frac{Q}{2r}e^{F+G}\phi-iA_0\left[\frac{(D-2)}{2}\phi+r\frac{\partial\phi}{\partial r}\right]+\frac{r}{2}e^{F+G}\frac{\partial V(|\phi|)}{\partial \phi^*},
	\end{align}
	where $\hat{R}(\hat{\sigma})$ describes the $(D-2)$-spatial Ricci scalar of compact manifold. We assume that the $(D-2)$-spatial geometry has to be Einstein which implies that the Ricci scalar $\hat{R}(\hat{\sigma})$ is constant.
	
	Let us introduce a generalized ansatz function
	\begin{align}\label{h}
	h :=r\frac{\partial\phi}{\partial r}+\frac{(D-2)}{2}\phi,
	\end{align}
	and
	\begin{align}\label{htilde}
	\phi :=\tilde{h}= r^{-\frac{(D-2)}{2}}\int_{0}^{r}hs^{\frac{D-4}{2}}\mathrm{d}s,
	\end{align}
	such that we have
	\begin{align}\label{dphi}
	\frac{\partial\phi}{\partial r}=\frac{h}{r}-\frac{(D-2)}{2}\frac{1}{r}\tilde{h}.
	\end{align}
	Using (\ref{htilde}) and (\ref{dphi}), we represent the charge function (\ref{local charge}) in terms of $h$ as follows
	\begin{align}\label{Qh}
	Q(u,r)=i~Vol(\Sigma^{D-2}) \int_{0}^{r}\left(\tilde{h}^*h-\tilde{h}h^*\right)s^{D-3}~\mathrm{d}s.
	\end{align}
	
	From the $\{rr\}$ and $\{ij\}$ components of the Einstein equations (see Appendix \ref{Appendix0}), we obtain
	\begin{align}
	g:=e^{F+G}=&\exp \left[-\frac{8\pi}{(D-2)}\int_{r}^{\infty}\frac{1}{s}\left(h-\frac{(D-2)}{2}\tilde{h}\right)^2\mathrm{d}s\right],\label{g2}\\
	\tilde{g}:= e^{F-G}=&\frac{\hat{R}(\hat{\sigma})}{(D-2)}\bar{g}-\frac{(D-4)}{(D-2)}\frac{\hat{R}(\hat{\sigma})}{r^{D-3}}\int_{0}^{r}\bar{g}s^{D-4}\mathrm{d}s+\frac{16\pi  }{(D-2)}\frac{1}{r^{D-3}}\int_{0}^{r} gV(|\tilde{h}|)s^{D-2}\mathrm{d}s\nonumber\\
	&-\frac{2}{(D-2)Vol^2(\Sigma^{D-2})}\frac{1}{r^{D-3}}\int_{0}^{r}gQ^2s^{-(D-2)}\mathrm{d}s\label{g tilde},
	\end{align}
	where $\bar{g}:=\bar{g}(u,r)=\frac{1}{r}\int_{0}^{r}g(u,s)\mathrm{d}s$. In the rest of the paper, we use the definition of the mean value formula $\bar{f}(u,r):=\frac{1}{r}\int_{0}^{r}f(u,s)\mathrm{d}s$. Hence, for the function $f:=f(.,r)$, we have $\frac{\partial \bar{f}}{\partial r}=\frac{f-\bar{f}}{r}$. 
	
	Now we introduce a new operator
	\begin{align}\label{Operator}
	\mathcal{D}:=\frac{\partial}{\partial u} - \frac{\tilde{g}}{2}\frac{\partial}{\partial r},
	\end{align}
	which describes the derivative along the incoming light rays parametrized by $u$. Finally, using (\ref{Operator}) together with (\ref{htilde}) and (\ref{dphi}), we represent the equation (\ref{EMH}) as the following single first-order integro-differential equation
	\begin{align}\label{Dh higher}
	\mathcal{D}h=&\left[\frac{1}{2r}\left(\frac{\hat{R}(\hat{\sigma})}{(D-2)}g-\frac{(D-2)}{2}\tilde{g}\right)+\frac{8\pi g r}{(D-2)}V(|\tilde{h}|)-\frac{gQ^2}{(D-2)r^3}\right]
	\left[h-\frac{(D-2)}{2}\tilde{h}\right]\nonumber\\
	&-i\frac{Q}{2r}g\tilde{h}-iA_0h+\frac{gr}{2}\frac{\partial V(|\tilde{h}|)}{\partial \tilde{h}^*}.
	\end{align}
	This provides the higher dimensional nonlinear evolution equation of the EMH system with respect to $h$.
	
	\subsection{Challenges overcome in this paper}
	Our main focus in this paper is on the study of the existence
	of global classical solution to the EMH system in higher dimensions. In extending the theory to the higher dimensions, there are some significant modifications we made that differ from the results in four dimensional case in Chae \cite{Chae1} as follows:
	\begin{enumerate}
		\item \textbf{Generalization of the ansatz function $h$}\\
		We obtain the equation of motion of the EMH system in higher dimensions shown by equation (\ref{EMH}). To simplify this equation, we define a new ansatz function $h$ given by equation (\ref{h}) that is more general than previous works by \cite{Chris1,Chae1},
		which is nothing but an ordinary differential equation with respect to $\phi$ with the solution $\tilde{h}$ given by equation (\ref{htilde}). In this paper we use the generalized function $\tilde{h}$ rather than the mean value function $\bar{h}:=\frac{1}{r}\int_{0}^{r}h~\mathrm{d}s$ such as in four dimensional case \cite{Chris1,Chae1}. Thus, we need to modify the details of the entire calculation, which is of course much more complicated.
		
		\item \textbf{Additional estimates for $\tilde{g}$}\\
		In the case of higher dimensions, the function $\tilde{g}$ is written by the equation (\ref{g tilde}).
		We have a new problem where the second term of the right-hand side of (\ref{g tilde}) appears in higher dimensions. Thus, we need to provide additional estimates for all function related to $\tilde{g}$. This requires some new modification techniques that are different from the four dimensional case.
		
		\item \textbf{Decay properties in higher dimensions}\\
		Decay properties play a crucial role in distinguishing the properties of the solutions to the Einstein equations in the four and higher dimensions. We know that the nontrivial vacuum solutions of the Einstein equations in Minkowski spacetime have the same order decay near the spatial infinity ($1/r^{D-3}$) and near the null infinity ($1/r^{(D-2)/2}$) for $D = 4$, but the decay at the infinite infinity is slower than the decay at the spatial infinity for $D>4$ \cite{Hollands,Hollands2}. Furthermore, Chae \cite{Chae1} has proved that in four dimensions, the solution of the EMH system decays with a decay order of $1/r^2$. Motivated by \cite{Prisma,Prisma2}, we prove that the solution of the higher dimensional EMH system decays with a decay order of $1/r^{k-1}$ where $k>\frac{D}{2}$.
		
		\item \textbf{Definition of local mass and local charge in higher dimensions}\\
		We also study the properties of spacetime completeness for EMH systems in higher dimensions. Therefore, we give new definitions for the local mass function $M(u,r)$ and local charge function $Q(u,r)$ in higher dimensions, where the final local mass and the final local charge are positive numbers. We show that the spacetime is complete along the time-like lines outside the region determined by the final local mass and the final local charge.
	\end{enumerate}
	
	\subsection{Structure of the paper}
	The organization of this paper is the following. Section \ref{Sec1} provides the motivation, construction of the problem, and the challenges. We extend the method originally proposed by Chae \cite{Chae1} and Christodoulou \cite{Chris1} in four dimensions to higher dimensions. This extension is not trivial since we need some modifications. We start the construction by introducing the coordinate system in higher dimensional metric (\ref{metric}) of Bondi coordinates. Then we define the new ansatz function in (\ref{h}) which is more general than the four dimensional case. We use this function to transform the equation of motion into a single first-order evolution equation as the main equation of our problem. We also define the generalized metric function $g$ and $\tilde{g}$ in (\ref{g2}) and (\ref{g tilde}) respectively. One of the challenges we discovered was that the second term on the right-hand-side of the equation (\ref{g tilde}) only appears in higher dimensions. To solve this problem we need to compute the estimates of the additional term for all functions containing $\tilde{g}$. In section \ref{Sec2}, we provide the mathematical setup that we use to prove the existence of solutions to our problem. We prove the existence of local solution in section \ref{Sec3} and the existence of global solution in section \ref{Sec4}. We firtsly show that the solution decays polynomially. Next, we prove the contraction mapping in the appropriate Banach space. Finally, we show that the problem has the global classical solution. We also study the completeness properties of the spacetimes in section \ref{Sec5}. We introduce a  function analogous to the mass of $D\geq 4$ in Bondi coordinates. The completeness of spacetimes along the future has led time-like lines outward into a region determined by the final local mass and the final local charge. We put the detail calculations as the Appendix in section \ref{Sec6}.
	
	\section{Mathematical setup}\label{Sec2}
	In this section we provide mathematical setup that we use to prove the existence of solutions to our problems. 
	
	Let us define the map $h\mapsto \mathcal{F}(h)$ as the solution of (\ref{Dh higher}), with the initial condition
	\begin{align}
	\mathcal{F}(h)(0,r)=h(0,r).
	\end{align}
	Let $r(u)=\chi(u;r_0)$ be the solution of the characteristics equation
	\begin{equation} \label{PDB}
	\frac{\mathrm{d}r}{\mathrm{d}u}=-\frac{1}{2}\tilde{g}(u,r),~~~r(0)=r_0.
	\end{equation}
	From equation (\ref{PDB}) we obtain the characteristic function $r_1:=\chi(u_1;r_0)$ as follows
	\begin{equation}\label{kondisi awal r}
	r_1 = r_0 - \frac{1}{2}\int_{0}^{u_1} \tilde{g}(u,\chi(u;r_0))~\mathrm{d}u.
	\end{equation}
	Using this characteristics, we write (\ref{Dh higher}) as the integral equation
	\begin{align}\label{Fcurl}
	\mathcal{F}(u_1,r_1)=&h(0,r_0)\exp\left\{\int_{0}^{u_1} \left(\frac{1}{2r}\left[\frac{\hat{R}(\hat{\sigma})}{(D-2)}g-\frac{(D-2)}{2}\tilde{g}\right]+\frac{8\pi g r}{(D-2)}V(|\tilde{h}|)\right.\right.\nonumber\\
	&\left.\left.-\frac{gQ^2}{(D-2)r^3}-iA_0\right)_\chi\mathrm{d}u'\right\}+\int_{0}^{u_1} \exp\left\{\int_{u}^{u_1} \left(\frac{1}{2r}\left[\frac{\hat{R}(\hat{\sigma})}{(D-2)}g-\frac{(D-2)}{2}\tilde{g}\right]\right.\right.\nonumber\\
	&\left.\left.+\frac{8\pi g r}{(D-2)}V(|\tilde{h}|)-\frac{gQ^2}{(D-2)r^3}-iA_0\right)_\chi\mathrm{d}u'\right\}[f]_\chi\mathrm{d}u,
	\end{align}
	where
	\begin{align}\label{f}
	f:=&-\left\{\frac{(D-2)}{2}\left(\frac{1}{2r}\left[\frac{\hat{R}(\hat{\sigma})}{(D-2)}g-\frac{(D-2)}{2}\tilde{g}\right]+\frac{8\pi g r}{(D-2)}V(|\tilde{h}|)\right.\right.\nonumber\\
	&\left.\left.-\frac{gQ^2}{(D-2)r^3}\right)+\frac{iQ}{2r}g\right\}\tilde{h}+\frac{gr}{2}\frac{\partial V(|\tilde{h}|)}{\partial \tilde{h}^*}.
	\end{align}
	
	Next, we define
	\begin{equation}
	\mathcal{G}(u,r):=\frac{\partial \mathcal{F}}{\partial r}(u,r),
	\end{equation}
	satisfies
	\begin{align}\label{eq.G}
	\mathcal{D}\mathcal{G}=&\left[\frac{1}{2}\frac{\partial \tilde{g}}{\partial r}+\frac{1}{2r} \left(\frac{\hat{R}(\hat{\sigma})}{(D-2)}g-\frac{(D-2)}{2}\tilde{g}\right)+\frac{8\pi g r}{(D-2)}V(|\tilde{h}|)-\frac{gQ^2}{(D-2)r^3}-iA_0\right]\mathcal{G}\nonumber\\
	& + \left[\frac{1}{2r}\frac{\partial}{\partial r}\left(\frac{\hat{R}(\hat{\sigma})}{(D-2)}g-\frac{(D-2)}{2}\tilde{g}\right)-\frac{1}{2r^2}\left(\frac{\hat{R}(\hat{\sigma})}{(D-2)}g-\frac{(D-2)}{2}\tilde{g}\right)\right.\nonumber\\
	&+\frac{8\pi g r }{(D-2)} \frac{\partial V(|\tilde{h}|)}{\partial\tilde{h}^*}\frac{\partial\tilde{h}^*}{\partial r}+\frac{8\pi g V(|\tilde{h}|)}{(D-2)}+\frac{8\pi r V(|\tilde{h}|)}{(D-2)}\frac{\partial g}{\partial r}-\frac{Q^2}{(D-2)r^3}\frac{\partial g}{\partial r}\nonumber\\
	&\left.-\frac{2gQ}{(D-2)r^3}\frac{\partial Q}{\partial r}+\frac{3gQ^2}{(D-2)r^4}\right]\left(\mathcal{F}-\frac{(D-2)}{2}\tilde{h}\right)+\left[\frac{r}{2}\frac{\partial g}{\partial r}+\frac{g}{2}\right]\frac{\partial V(|\tilde{h}|)}{\partial\tilde{h}^*}\nonumber\\
	& -\left[\frac{(D-2)}{2}\left(\frac{1}{2r} \left(\frac{\hat{R}(\hat{\sigma})}{(D-2)}g-\frac{(D-2)}{2}\tilde{g}\right)+\frac{8\pi g r V(|\tilde{h}|)}{(D-2)}\right)+\frac{gQ^2}{2r^3}-\frac{igQ}{2r}\right]\frac{\partial \tilde{h}}{\partial r}\nonumber\\
	&+\frac{gr}{2}\frac{\partial^2 V(|\tilde{h}|)}{(\partial\tilde{h}^*)^2}\frac{\partial\tilde{h}^*}{\partial r}-i\frac{\partial A_0}{\partial r}\mathcal{F}+\left[\frac{igQ}{2r^2}-\frac{ig}{2r}\frac{\partial Q}{\partial r}-\frac{iQ}{2r}\frac{\partial g}{\partial r}\right]\tilde{h},
	\end{align}
	with the initial condition $\mathcal{G}(0,r_0)=\frac{\partial h}{\partial r}(0,r_0)$. Using the characteristics as previously, we write the equation (\ref{eq.G}) as the integral equation
	\begin{align}\label{Gcurl}
	\mathcal{G}(u_1,r_1)=&\frac{\partial h}{\partial r}(0,r_0)\exp\left\{\int_{0}^{u_1}\left[\frac{1}{2}\frac{\partial \tilde{g}}{\partial r}+\frac{1}{2r}\left[\frac{\hat{R}(\hat{\sigma})}{(D-2)}g-\frac{(D-2)}{2}\tilde{g}\right]\right.\right.\nonumber\\
	&\left.\left.-\frac{8\pi gr}{(D-2)} V(|\tilde{h}|)-\frac{Q^2}{(D-2)r^3}g-iA_0\right]_\chi\mathrm{d}u \right\}\nonumber\\
	&+\int_{0}^{u_1}\exp\left\{\int_{u}^{u_1}\left[\frac{1}{2}\frac{\partial \tilde{g}}{\partial r}+\frac{1}{2r}\left[\frac{\hat{R}(\hat{\sigma})}{(D-2)}g-\frac{(D-2)}{2}\tilde{g}\right]\right.\right.\nonumber\\
	&\left.\left.-\frac{8\pi gr}{(D-2)} V(|\tilde{h}|)-\frac{Q^2}{(D-2)r^3}g-iA_0\right]_\chi\mathrm{d}u' \right\}[f_1]_\chi\mathrm{d}u,
	\end{align}
	where
	\begin{align} \label{f1}
	f_1 :=& \left[\frac{1}{2r}\frac{\partial}{\partial r}\left(\frac{\hat{R}(\hat{\sigma})}{(D-2)}g-\frac{(D-2)}{2}\tilde{g}\right)-\frac{1}{2r^2}\left(\frac{\hat{R}(\hat{\sigma})}{(D-2)}g-\frac{(D-2)}{2}\tilde{g}\right)\right.\nonumber\\
	&+\frac{8\pi g r }{(D-2)} \frac{\partial V(|\tilde{h}|)}{\partial\tilde{h}^*}\frac{\partial\tilde{h}^*}{\partial r}+\frac{8\pi g V(|\tilde{h}|)}{(D-2)}+\frac{8\pi r V(|\tilde{h}|)}{(D-2)}\frac{\partial g}{\partial r}-\frac{Q^2}{(D-2)r^3}\frac{\partial g}{\partial r}\nonumber\\
	&\left.-\frac{2gQ}{(D-2)r^3}\frac{\partial Q}{\partial r}+\frac{3gQ^2}{(D-2)r^4}\right]\left(\mathcal{F}-\frac{(D-2)}{2}\tilde{h}\right)+\left[\frac{r}{2}\frac{\partial g}{\partial r}+\frac{g}{2}\right]\frac{\partial V(|\tilde{h}|)}{\partial\tilde{h}^*}\nonumber\\
	& -\left[\frac{(D-2)}{2}\left(\frac{1}{2r} \left(\frac{\hat{R}(\hat{\sigma})}{(D-2)}g-\frac{(D-2)}{2}\tilde{g}\right)+\frac{8\pi g r V(|\tilde{h}|)}{(D-2)}\right)+\frac{gQ^2}{2r^3}-\frac{igQ}{2r}\right]\frac{\partial \tilde{h}}{\partial r}\nonumber\\
	&+\frac{gr}{2}\frac{\partial^2 V(|\tilde{h}|)}{(\partial\tilde{h}^*)^2}\frac{\partial\tilde{h}^*}{\partial r}-i\frac{\partial A_0}{\partial r}\mathcal{F}+\left[\frac{igQ}{2r^2}-\frac{ig}{2r}\frac{\partial Q}{\partial r}-\frac{iQ}{2r}\frac{\partial g}{\partial r}\right]\tilde{h}.
	\end{align}
	
	Then, we define $g_l:=g(h_l)$, $\mathcal{F}_l:=\mathcal{F}(h_l)$, $\mathcal{G}_l:=\mathcal{G}(h_l)$, ${A_0}_l:=A_0(h_l)$ for $l=1,2$, and 
	\begin{align}\label{Def Theta}
	\Theta:=\mathcal{F}(h_1)-\mathcal{F}(h_2),
	\end{align}
	satisfies
	\begin{align}
	\mathcal{D}\Theta=&\left[\frac{1}{2r}\left[\frac{\hat{R}(\hat{\sigma})}{(D-2)}g_1-\frac{(D-2)}{2}\tilde{g}_1\right]+\frac{8\pi r g_2}{(D-2)} V(|\tilde{h}_1|)-\frac{g_1Q_2^2}{(D-2)r^3}-i{A_0}_1\right]\Theta\nonumber\\
	&+\frac{1}{2}\left(\tilde{g}_1-\tilde{g}_2\right)\mathcal{G}_2-\frac{1}{2r}\left[\frac{\hat{R}(\hat{\sigma})}{(D-2)}g_1-\frac{(D-2)}{2}\tilde{g}_1\right]\frac{(D-2)}{2}(\tilde{h}_1 -\tilde{h}_2)\nonumber\\
	&+\frac{1}{2r}\left(\frac{\hat{R}(\hat{\sigma})}{(D-2)}(g_1-g_2)-\frac{(D-2)}{2}(\tilde{g}_1-\tilde{g}_2)\right)\mathcal{F}_2\nonumber\\
	&-\frac{1}{2r}\left(\frac{\hat{R}(\hat{\sigma})}{(D-2)}(g_1-g_2)-\frac{(D-2)}{2}(\tilde{g}_1-\tilde{g}_2)\right)\frac{(D-2)}{2}\tilde{h}_2\nonumber\\
	&+\frac{8\pi r}{(D-2)}(g_1-g_2)\mathcal{F}_1V(|\tilde{h}_1|)-4\pi r(g_1-g_2)\tilde{h}_1 V(|\tilde{h}_1|)\nonumber\\
	&+\frac{8\pi r g_2}{(D-2)} \left(V(|\tilde{h}_1|)-V(|\tilde{h}_2|)\right)\mathcal{F}_2-4\pi r g_2 (\tilde{h}_1 - \tilde{h}_2)V(|\tilde{h}_2|)\nonumber\\
	&-4\pi r g_2\tilde{h}_1 \left(V(|\tilde{h}_1|)-V(|\tilde{h}_2|)\right)+\frac{r}{2}(g_1-g_2)\frac{\partial V(|\tilde{h}_1|)}{\partial \tilde{h}_1^*}\nonumber\\
	&+\frac{rg_2}{2}\left(\frac{\partial V(|\tilde{h}_1|)}{\partial \tilde{h}_1^*}-\frac{\partial V(|\tilde{h}_2|)}{\partial \tilde{h}_2^*}\right)-\frac{g_1\left(Q_1^2-Q_2^2\right)}{(D-2)r^3}\left(\mathcal{F}_1-\frac{(D-2)}{2}\tilde{h}_1\right)\nonumber\\
	&- \frac{g_1Q_2^2}{2r^3}(\tilde{h}_1-\tilde{h}_2)-\frac{(g_1-g_2)}{(D-2)r^3}Q_2^2\left(\mathcal{F}_2-\frac{(D-2)}{2}\tilde{h}_2\right)-i\frac{(Q_1-Q_2)}{2r}g_1\tilde{h}_1\nonumber\\
	&-i\frac{Q_2}{2r}(g_1-g_2)\tilde{h}_1-i\frac{Q_2g_2}{2r}(\tilde{h}_1-\tilde{h}_2)-i\mathcal{F}_2({A_0}_1-{A_0}_2).
	\end{align}
	
	Let us consider a new characteristics equation
	\begin{align}
	\frac{\mathrm{d}r}{\mathrm{d}u}=-\frac{\tilde{g}_1}{2}(\chi_1(u,r),u);~~~r(0)=r_0.
	\end{align}
	Setting $h_1(0,r_1)=h_2(0,r_2)$. Hence we write (\ref{Def Theta}) as the integral equation
	\begin{align}\label{Theta}
	\Theta(u_1,r_1)=&\int_{0}^{u_1} \exp\left\{\int_{u}^{u_1} \left[\frac{1}{2r}\left[\frac{\hat{R}(\hat{\sigma})}{(D-2)}g_1-\frac{(D-2)}{2}\tilde{g}_1\right]+\frac{8\pi r g_2}{(D-2)} V(|\tilde{h}_1|)\right.\right.\nonumber\\
	&\left.\left.-\frac{g_1Q_2^2}{(D-2)r^3}-i{A_0}_1\right]_{\chi_1}\mathrm{d}u'\right\}[\tilde{\varphi}]_{\chi_1}\mathrm{d}u\;,
	\end{align}
	where
	\begin{align} \label{phi tilde}
	\tilde{\varphi}:=&\frac{1}{2}\left(\tilde{g}_1-\tilde{g}_2\right)\mathcal{G}_2-\frac{1}{2r}\left[\frac{\hat{R}(\hat{\sigma})}{(D-2)}g_1-\frac{(D-2)}{2}\tilde{g}_1\right]\frac{(D-2)}{2}(\tilde{h}_1 -\tilde{h}_2)\nonumber\\
	&+\frac{1}{2r}\left(\frac{\hat{R}(\hat{\sigma})}{(D-2)}(g_1-g_2)-\frac{(D-2)}{2}(\tilde{g}_1-\tilde{g}_2)\right)\mathcal{F}_2\nonumber\\
	&-\frac{1}{2r}\left(\frac{\hat{R}(\hat{\sigma})}{(D-2)}(g_1-g_2)-\frac{(D-2)}{2}(\tilde{g}_1-\tilde{g}_2)\right)\frac{(D-2)}{2}\tilde{h}_2\nonumber\\
	&+\frac{8\pi r}{(D-2)}(g_1-g_2)\mathcal{F}_1V(|\tilde{h}_1|)-4\pi r(g_1-g_2)\tilde{h}_1 V(|\tilde{h}_1|)\nonumber\\
	&+\frac{8\pi r g_2}{(D-2)} \left(V(|\tilde{h}_1|)-V(|\tilde{h}_2|)\right)\mathcal{F}_2-4\pi r g_2 (\tilde{h}_1 - \tilde{h}_2)V(|\tilde{h}_2|)\nonumber\\
	&-4\pi r g_2\tilde{h}_1 \left(V(|\tilde{h}_1|)-V(|\tilde{h}_2|)\right)+\frac{r}{2}(g_1-g_2)\frac{\partial V(|\tilde{h}_1|)}{\partial \tilde{h}_1^*}\nonumber\\
	&+\frac{rg_2}{2}\left(\frac{\partial V(|\tilde{h}_1|)}{\partial \tilde{h}_1^*}-\frac{\partial V(|\tilde{h}_2|)}{\partial \tilde{h}_2^*}\right)-\frac{g_1\left(Q_1^2-Q_2^2\right)}{(D-2)r^3}\left(\mathcal{F}_1-\frac{(D-2)}{2}\tilde{h}_1\right)\nonumber\\
	&- \frac{g_1Q_2^2}{2r^3}(\tilde{h}_1-\tilde{h}_2)-\frac{(g_1-g_2)}{(D-2)r^3}Q_2^2\left(\mathcal{F}_2-\frac{(D-2)}{2}\tilde{h}_2\right)-i\frac{(Q_1-Q_2)}{2r}g_1\tilde{h}_1\nonumber\\
	&-i\frac{Q_2}{2r}(g_1-g_2)\tilde{h}_1-i\frac{Q_2g_2}{2r}(\tilde{h}_1-\tilde{h}_2)-i\mathcal{F}_2({A_0}_1-{A_0}_2).
	\end{align}
	
	\section{Local existence}\label{Sec3}
	This section is devoted to prove the existence of local classical solution to (\ref{Dh higher}). Local existence describes that the solution of the system (\ref{Dh higher}) exists for a finite time $u_0$.
	
	Setting $k>\frac{D}{2}$, where $k\in\mathbb{R}^+$.  We introduce a function space
	\begin{align}\label{space Xhat}
	\hat{X} :=\{h\in C^1([0,u_0] \times [0,\infty) )\; | \;\|h\|_{\hat{X}} < \infty\}\:, 
	\end{align}
	equipped with the norm
	\begin{align} \label{hx hat}
	\|h\|_{\hat{X}} :=\sup_{u\in[0,u_0]} \sup_{r\geq 0}\left\{ (1+r+u)^{k-1}|h(u,r)| + (1+r+u)^{k}\left|\frac{\partial h}{\partial r}(u,r)\right|  \right\}\:.
	\end{align}
	Then, we introduce
	\begin{align}
	\hat{X}_0:=\{h\in C^1([0,u_0])\; | \;\|h\|_{\hat{X}_0} < \infty\}\:,
	\end{align}
	with the norm
	\begin{align}
	\|h\|_{\hat{X}_0} :=\sup_{r\geq 0}\left\{ (1+r)^{k-1}|h(r)| + (1+r)^{k}\left|\frac{\partial h}{\partial r}(r)\right|  \right\}\:.
	\end{align}
	We also introduce the space $\hat{Y}$ containing $\hat{X}$ as follows
	\begin{align}\label{space Yhat}
	\hat{Y}:=\{h\in C^1([0,u_0] \times [0,\infty) )\; | \; h(0,r) = h_0(r), \|h\|_{\hat{Y}} < \infty\}\:,
	\end{align}
	with
	\begin{align}\label{norm Yhat}
	\|h\|_{\hat{Y}} :=\sup_{u\in[0,u_0]} \sup_{r\geq 0}\left\{ (1+r+u)^{k-1}|h(u,r)|\right\}\:,
	\end{align}
	The space function $\hat{X}, \hat{X}_0,$ and $\hat{Y}$ are the Banach spaces. Let us denote $\|h\|_{\hat{X}}:=\hat{x}$, $\|h(0,.)\|_{\hat{X}_0}:=\hat{d}$, and $\|h_1-h_2\|_{\hat{Y}}:=\hat{y}$.
	
	The existence of local solutions of \eqref{Dh higher} can be stated in the following Theorem.
	\begin{theorem}\label{Theorem Local}
		Let $D\geq 4$. Given a positive constant $(D-2)$-spatial Ricci scalar of compact manifold $\hat{R}(\hat{\sigma})$.
		Suppose that $\hat{X}$ and $\hat{Y}$ are the function spaces defined by (\ref{space Xhat}) and (\ref{space Yhat}) respectively. For a given an initial data $h(0,r)\in C^1[0,\infty)$ such that $h(0,r)=O(r^{-{(k-1)}})$ and $\frac{\partial h}{\partial r}(0,r)=O(r^{-k})$ as $r\rightarrow \infty$, there exists a $u_0>0$ and a local classical solution of the equation (\ref{Dh higher}) satisfies
		\begin{align}
		h(u,r)\in C^1 ([0,u_0]\times [0,\infty)),
		\end{align}
		for $p\in[k,\infty)$ and $k>\frac{D}{2}$, where $p,k\in\mathbb{R}^+$.
	\end{theorem}
	
	To prove Theorem \ref{Theorem Local} we establish two Lemmas.
	\begin{lemma}\label{Lemma local 1}
		Let $D\geq 4$. Setting $p\in[k,\infty)$ and $k>\frac{D}{2}$, where $p,k\in\mathbb{R}^+$. For any $\mathcal{L}_1(\hat{x})>\hat{d}$, there exists $\delta(\hat{x},\hat{d})>0$ such that if $u_0<\delta$, the map $\mathcal{F}$ is contained in the closed ball of radius $\hat{x}$ in the space $\hat{X}$. 
	\end{lemma}
	\begin{proof}
		We firstly calculate
		\begin{align}\label{estimasi htilde local}
		|\tilde{h}|\leq& r^{-\frac{(D-2)}{2}}\int_{0}^{r}\frac{\|h\|_{\hat{X}}}{(1+s+u)^{2}}s^{\frac{(D-4)}{2}}\mathrm{d}s\nonumber\\
		\leq&\frac{2\hat{x}}{(2k-D)}\frac{1}{(1+u)^{\frac{(2k-D)}{2}}(1+r+u)^{\frac{(D-2)}{2}}}.
		\end{align}
		Since $|\tilde{h}|$ is always positive, then we have $k>\frac{D}{2}$.
		
		From (\ref{estimasi htilde local}), we have
		\begin{align}\label{h-h local}
		\left|h-\frac{(D-2)}{2}\tilde{h}\right|\leq |h| + \left|\frac{(D-2)}{2}\tilde{h}\right|\leq \frac{C\hat{x}}{(1+u)^{\frac{(2k-D)}{2}}(1+r+u)^{\frac{(D-2)}{2}}}.
		\end{align}
		Using the main value formula we obtain
		\begin{align}\label{g-gbar k local}
		|(g-\bar{g})(u,r)|\leq& \frac{1}{r}\int_{0}^{r}|g(u,r)-g(u,r')|\mathrm{d}r'
		\nonumber\\
		\leq& \frac{8\pi}{(D-2)}\frac{1}{r}\int_{0}^{r} \int_{r'}^{r}\frac{1}{s}\left|h-\frac{(D-2)}{2}\tilde{h}\right|^2\mathrm{d}s\mathrm{d}r'
		\nonumber\\
		\leq& \frac{8\pi \hat{x}^2}{(D-3)(D-2)^2}\frac{1}{(1+u)^{2k-3}(1+r+u)}.
		\end{align}
		Substituting (\ref{h-h local}) into (\ref{g2}) yields
		\begin{align}
		|g (u,r)|\geq\exp \left[-\frac{8\pi \hat{x}^2}{(D-2)^2(1+u)^{2k-D}(1+r+u)^{D-2}}\right].
		\end{align}
		Then we use triangle inequality to produce
		\begin{align}\label{estimate gbar local}
		|\bar{g}|\geq |g| + |g-\bar{g}|\geq\frac{8\pi \hat{x}^2}{(D-3)(D-2)^2}\frac{1}{(1+u)^{2k-3}(1+r+u)},
		\end{align}
		such that we have
		\begin{align}
		\frac{(D-4)}{2}\frac{\hat{R}(\hat{\sigma})}{r^{D-3}}\int_{0}^{r}|\bar{g}|s^{D-4}\mathrm{d}s\leq \frac{C\hat{x}^2}{(1+u)^{2k-3}(1+r+u)}.
		\end{align}
		Setting $p\in[k,\infty)$, we obtain
		\begin{align}
		\frac{8\pi}{r^{D-3}}\int_{0}^{r} gs^{D-2}|V(|\tilde{h}|)|\mathrm{d}s\leq \frac{C\hat{x}^{p+1}}{(1+u)^{k^2-D}(1+r+u)^{D-3}}.
		\end{align}
		In the view of (\ref{Qh}) we have
		\begin{align}\label{estimate Q local}
		|Q(u,r)|=&\left|Vol(\Sigma^{D-2})i \int_{0}^{r}\left(\tilde{h}^*h-\tilde{h}h^*\right)s^{D-3}~\mathrm{d}s\right|\nonumber\\
		\leq&2Vol(\Sigma^{D-2})\int_{0}^{r}|\tilde{h}||h|s^{D-3}~\mathrm{d}s\nonumber\\
		\leq& \frac{8Vol (\Sigma^{D-2})}{(2k-D)^2}\frac{\hat{x}^2 r^{\frac{(2k-D)}{2}}}{(1+u)^{2k-D}(1+r+u)^{\frac{(2k-D)}{2}}}.
		\end{align}
		Hence, it is straightforward to get
		\begin{align}\label{estimate integral Q local}
		\frac{1}{Vol^2(\Sigma^{D-2})}\frac{1}{r^{D-3}}\int_{0}^{r}\left|\frac{gQ^2}{s^{D-2}}\right|\mathrm{d}s\leq\frac{64}{(D-3)(2k-D)^4}\frac{\hat{x}^4}{(1+u)^{4k-D-3}(1+r+u)^{D-3}},
		\end{align}
		and
		\begin{align}\label{estimate A0 local}
		|iA_0|\leq \frac{1}{Vol(\Sigma^{D-2})} \int_{0}^{r}\frac{g|Q|}{s^{D-2}}\mathrm{d}s\leq& \frac{8\hat{x}^2}{(2k-D)^2}\frac{1}{(1+u)^{2k-D}}\int_{0}^{r}\frac{1}{(1+s+u)^{D-2}}\mathrm{d}s\nonumber\\
		\leq&\frac{8\hat{x}^2}{(D-3)(2k-D)^2}\frac{r^{D-3}}{(1+u)^{2k-3}(1+r+u)^{D-3}}.
		\end{align}
		Then, from (\ref{g tilde}) we obtain
		\begin{align}\label{g - gtilde local}
		\left|\frac{\hat{R}(\hat{\sigma})}{(D-2)}g-\frac{(D-2)}{2}\tilde{g}\right|\leq&|g-\bar{g}|	+\frac{(D-4)}{2}\frac{\hat{R}(\hat{\sigma})}{r^{D-3}}\int_{0}^{r}|\bar{g}|s^{D-4}\mathrm{d}s+\frac{8\pi}{r^{D-3}}\int_{0}^{r} gs^{D-2}|V(|\tilde{h}|)|\mathrm{d}s\nonumber\\
		&+	\frac{1}{Vol^2(\Sigma^{D-2})}\frac{1}{r^{D-3}}\int_{0}^{r}\left|\frac{gQ^2}{s^{D-2}}\right|\mathrm{d}s\nonumber\\
		\leq& \frac{C(\hat{x}^2+\hat{x}^4+\hat{x}^{p+1})}{(1+u)^{2k-3}(1+r+u)}.
		\end{align}
		Combining (\ref{potential}), (\ref{htilde}), (\ref{estimate Q local}), and (\ref{g - gtilde local}), we get
		\begin{align}\label{estimate f local}
		|f|\leq \frac{C(\hat{x}^3+\hat{x}^5+\hat{x}^p+\hat{x}^{p+2})}{(1+u)^\frac{(2k-D)k}{2}(1+r+u)^{\frac{D}{2}}}.
		\end{align}	
		
		From the definition $\|h(0,.)\|_{\hat{X}_0}:=\hat{d}$, we obtain
		\begin{align}\label{h local}
		|h(0,r)|\leq \frac{C\hat{d}}{(1+r)^{k-1}}.
		\end{align}
		
		Suppose that $h(u,r)$ is defined in $\mathcal{I}=[0,\delta]$ that contains $[0,u_0]$ such that $\delta\in \mathcal{I}$ but $\delta\notin [0,u_0]$. Therefore, we write down the estimate for exponential term of (\ref{Fcurl}) from $u\in[0,u_0]$ up to $u'=\delta>u_0$ as follows
		\begin{align}\label{delta1}
		&\int_{0}^{\delta} \left[\frac{1}{2r}\left|\frac{\hat{R}(\hat{\sigma})}{(D-2)}g-\frac{(D-2)}{2}\tilde{g}\right|+\frac{8\pi g r|V(|\tilde{h}|)|}{(D-2)}+\frac{|Q^2|g}{(D-2)r^3}+|iA_0|\right]\mathrm{d}u\nonumber\\
		&\leq C(\hat{x}^2+\hat{x}^4+\hat{x}^{p+1})\delta.
		\end{align}
		Combining (\ref{estimate f local}), (\ref{h local}), and (\ref{delta1}), we represent the following estimate
		\begin{align}\label{Fcurl local}
		|\mathcal{F}(\delta,r)|\leq\frac{\hat{C}_1(\hat{d}+\hat{x}^3+\hat{x}^5+\hat{x}^p+\hat{x}^{p+2})\delta\exp\left[\hat{C}_2(\hat{x}^2+\hat{x}^4+\hat{x}^{p+1})\delta\right]}{(1+r)^{k-1}},
		\end{align}
		for any $\hat{C}_1, \hat{C}_2 >0$ depends on $k$ and dimensions $D$.
		
		From the mean value formula we have $\frac{\partial\bar{g}}{\partial r}=\frac{g-\bar{g}}{r}$. Then, taking differentiation of (\ref{g tilde}) with respect to $r$ we get
		\begin{align}\label{dgtilde}
		\left|\frac{\partial \tilde{g}}{\partial r}\right|\leq&\frac{\hat{R}(\hat{\sigma})}{(D-2)}\frac{|g-\bar{g}|}{r}+\frac{\hat{R}(\hat{\sigma})(D-4)(D-3)}{(D-2)r^{D-2}}\int_{0}^{r}|\bar{g}|s^{D-4}\mathrm{d}s+\frac{\hat{R}(\hat{\sigma})(D-4)}{(D-2)}\frac{|\bar{g}|}{r}\nonumber\\
		&+\frac{(D-3)16\pi}{(D-2)r^{D-2}}\int_{0}^{r} \left|g\right|s^{D-2}\left|V(|\tilde{h}|)\right|\mathrm{d}s+\frac{16\pi |g|r\left|V(|\tilde{h}|)\right|}{(D-2)}\nonumber\\
		&+\frac{2(D-3)}{(D-2)Vol^2(\Sigma^{D-2})}\frac{1}{r^{D-2}}\int_{0}^{r}\left|\frac{gQ^2}{s^{D-2}}\right|\mathrm{d}s\nonumber\\
		&+\frac{2}{(D-2)Vol^2(\Sigma^{D-2})}\left|\frac{gQ^2}{r^{2D-5}}\right|.
		\end{align}
		Thus, by (\ref{potential}), (\ref{g2}), (\ref{g-gbar k local}), (\ref{estimate gbar local}), and (\ref{estimate Q local}) we obtain
		\begin{align}
		\left|\frac{1}{2r}\left(\frac{\hat{R}(\hat{\sigma})}{(D-2)}\frac{\partial g}{\partial r}-\frac{(D-2)}{2}\frac{\partial\tilde{g}}{\partial r}\right)\right|\leq\frac{C(\hat{x}^2+\hat{x}^4+\hat{x}^{p+1})}{(1+u)^{2k-D}(1+r+u)^3}.
		\end{align}
		In the view of (\ref{htilde}), we have
		\begin{align}
		\left|\frac{\partial \tilde{h}}{\partial r}\right|= \frac{1}{r}\left|h-\frac{D-2}{2}\tilde{h}\right|\leq\frac{C\hat{x}}{(1+u)^{\frac{(2k-D)}{2}}(1+r+u)^{\frac{D}{2}}}.
		\end{align}
		Furthermore, from the estimate (\ref{estimate Q local}) we calculate
		\begin{align}
		\left|\frac{\partial Q}{\partial r}\right|\leq&\left|Vol(\Sigma^{D-2})i \left(\tilde{h}^*h-\tilde{h}h^*\right)r^{D-3}\right|\nonumber\\
		\leq& 2Vol(\Sigma^{D-2}) \left|\tilde{h}\right|\left|h\right|r^{D-3}\nonumber\\
		\leq&\frac{C\hat{x}^2}{(1+u)^{\frac{(2k-D)}{2}}(1+r+u)^{\frac{2k-D+2}{2}}}.
		\end{align}
		Finally, we obtain
		\begin{align}\label{estimate f1 local}
		|f_1|\leq\frac{C(\hat{x}^2+\hat{x}^4+\hat{x}^6+\hat{x}^{p+1}+\hat{x}^{p+3})}{(1+u)^{2k-D}(1+r+u)^{D-2}}|\mathcal{F}|+\frac{C(\hat{x}^3+\hat{x}^5+\hat{x}^7+\hat{x}^p+\hat{x}^{p+2}+\hat{x}^{p+4})}{(1+u)^{2k-D}(1+r+u)^\frac{(D+2)}{2}}.
		\end{align}
		
		Then, we calculate
		\begin{align}
		\left|\frac{\partial h}{\partial r}(0,r) \right|\leq \frac{C\hat{d}}{(1+r)^k}.
		\end{align}
		
		Using the similar calculation as the estimate (\ref{delta1}), we obtain
		\begin{align}
		&\int_{0}^{\delta}\left[\frac{1}{2}\left|\frac{\partial\tilde{g}}{\partial r}\right|+\frac{1}{2r}\left|\left(\frac{\hat{R}(\hat{\sigma})}{(D-2)}g-\frac{(D-2)}{2}\tilde{g}\right) \right|+\frac{8\pi g r|V(|\tilde{h}|)|}{(D-2)}+\frac{g|Q^2|}{(D-2)r^3}+|iA_0|\right]\mathrm{d}u\nonumber\\
		&\leq C\left(\hat{x}^2 + \hat{x}^{4} + \hat{x}^{p+1} \right)\delta.
		\end{align}
		Therefore, we write the estimate for (\ref{Gcurl}) as follows
		\begin{align}\label{Gcurl local}
		\left|\mathcal{G}(\delta,r)\right|
		\leq&\frac{\hat{C}_3(\hat{d}+\hat{x}^3+\hat{x}^{5}+\hat{x}^7+\hat{x}^p+\hat{x}^{p+2}+\hat{x}^{p+4})}{(1+r)^k}\nonumber\\
		&\times(1+\hat{x}^2+\hat{x}^{4}+\hat{x}^6+\hat{x}^{p+1}+\hat{x}^{p+3})\delta \exp\left[\hat{C}_4(\hat{x}^2+\hat{x}^{4}+\hat{x}^{p+1})\delta\right],
		\end{align}
		for any $\hat{C}_3, \hat{C}_4 >0$ depends on $k$ and dimensions $D$.
		
		In the view of norm for the space $\hat{X}$, we have
		\begin{align}
		\|\mathcal{F}\|_{\hat{X}}\leq&\hat{C}_5(\hat{d}+\hat{x}^3+\hat{x}^5+\hat{x}^{7}+\hat{x}^p+\hat{x}^{p+2}+\hat{x}^{p+4})\nonumber\\
		&\times(1+\hat{x}^2+\hat{x}^{4}+\hat{x}^6+\hat{x}^{p+1}+\hat{x}^{p+3})\delta\exp\left[\hat{C}_6(\hat{x}^2+\hat{x}^{4}+\hat{x}^{p+1})\delta\right],
		\end{align}
		for any $\hat{C}_5, \hat{C}_6 >0$ depends on $k$ and dimensions $D$.
		
		Let us introduce a function
		\begin{align}
		\mathcal{L}_1(\hat{x}):=&\frac{\hat{x}\exp\left[-\hat{C}_6(\hat{x}^2+\hat{x}^{4}+\hat{x}^{p+1})\delta\right]}{\hat{C}_5(1+\hat{x}^2+\hat{x}^{4}+\hat{x}^6+\hat{x}^{p+1}+\hat{x}^{p+3})\delta}\nonumber\\
		&-\left(\hat{x}^3+\hat{x}^5+\hat{x}^{7}+\hat{x}^p+\hat{x}^{p+2}+\hat{x}^{p+4}\right).
		\end{align}
		We have $\mathcal{L}_1(0)=0$, $\mathcal{L}_1'(0)>0$, and $\mathcal{L}_1(\hat{x})\rightarrow-\infty$ as $\hat{x}\rightarrow\infty$. There exists $\hat{x}_0\in(0,\hat{x}_1)$ such that $\mathcal{L}_1(\hat{x})$ is monotonically increasing on $[0,\hat{x}_0]$. Thus, for every $\hat{x}\in(0,\hat{x}_0)$, we conclude that $\|\mathcal{F}\|_{\hat{X}}\leq \hat{x}$ such that the map $\mathcal{F}$ is contained in the closed ball of radius $\hat{x}$ in the space $\hat{X}$ if $\hat{d}<\mathcal{L}_1(\hat{x})$. The proof is finished.
	\end{proof}
	
	\begin{lemma}\label{Lemma local 2}
		Let $D\geq 4$. Setting $p\in[k,\infty)$ and $k>\frac{D}{2}$, where $p,k\in\mathbb{R}^+$. There exists $\delta:=\delta(\hat{x})>0$, such that if $u_0<\delta$, the map  $\mathcal{F}$ contracts in $\hat{Y}$.
	\end{lemma}
	\begin{proof}
		Suppose that equation (\ref{Dh higher}) has two different solutions, namely $h_1, h_2 \in \hat{X}$, that are defined in $\mathcal{I}=[0,\delta]$ and contains $[0,u_0]$ such that $\delta\in \mathcal{I}_2$ but $\delta\notin [0,u_0]$. We assume
		\begin{align}
		\max\{\|h_1\|_{\hat{X}},\|h_2\|_{\hat{X}}\}<\hat{x}.
		\end{align}
		
		We have
		\begin{align}\label{h1-h2 local}
		|\tilde{h}_1 - \tilde{h}_2|\leq\frac{2\hat{y}}{(2k-D)}\frac{1}{(1+u)^{\frac{(2k-D)}{2}}(1+r+u)^{\frac{(D-2)}{2}}}.
		\end{align}
		Using the above estimate we obtain
		\begin{align}
		|h_1-h_2-(\tilde{h}_1-\tilde{h}_2)|\leq |h_1-h_2|+|\tilde{h}_1-\tilde{h}_2|
		\leq\frac{C\hat{y}}{(1+u)^{\frac{(2k-D)}{2}}(1+r+u)^{\frac{(D-2)}{2}}},
		\end{align}
		and
		\begin{align}\label{h-htilde local}
		\left||h_1-\tilde{h}_1|^2-|h_2-\tilde{h}_2|^2\right|\leq& \left|(h_1-h_2)-(\tilde{h}_1-\tilde{h}_2) \right|\left(|h_1-\tilde{h}_1|+|h_2-\tilde{h}_2|\right)\nonumber\\
		\leq& \frac{C\hat{x}\hat{y}}{(1+u)^{2k-D}(1+r+u)^{D-2}}.
		\end{align}
		Then, we use the relation
		\begin{align}
		\tilde{h}_1^{p+1} - \tilde{h}_2^{p+1}=(\tilde{h}_1-\tilde{h}_2)\int_{0}^{1}\left(t\tilde{h}_1+(1-t)\tilde{h}_2\right)\left|t\tilde{h}_1+(1-t)\tilde{h}_2\right|^{p-1}\mathrm{d}t,
		\end{align}
		such that we have
		\begin{align}\label{V1-V2 local}
		\left|V(|\tilde{h}_1|)-V(|\tilde{h}_2|)\right|\leq&\left|\tilde{h}_1^{p+1}-\tilde{h}_2^{p+1}\right|\leq |\tilde{h}_1-\tilde{h}_2|\left(|\tilde{h}_1|+|\tilde{h}_2|\right)^p\nonumber\\
		\leq&\frac{2^{2p+1}\hat{x}^p\hat{y}}{(2k-D)^{k+1}(1+u)^{\frac{(2k-D)(k+1)}{2}}(1+r+u)^{\frac{(D-2)(k+1)}{2}}}.
		\end{align}
		Furthermore, we calculate
		\begin{align}
		\left|\frac{\partial V(|\tilde{h}_1|)}{\partial \tilde{h}_1^*}-\frac{\partial V(|\tilde{h}_2|)}{\partial \tilde{h}_2^*}\right|
		\leq& K_0|\tilde{h}_1-\tilde{h}_2|(|\tilde{h}_1|^{p-1}+|\tilde{h}_2|^{p-1})\nonumber\\
		\leq&\frac{C \hat{x}^{p-1}\hat{y}}{(1+u)^{\frac{(2k-D)k}{2}}(1+r+u)^{\frac{(D-2)k}{2}}}.
		\end{align}
		
		Next we calculate
		\begin{align}\label{Q1-Q2 local}
		|Q_1-Q_2|\leq&Vol(\Sigma^{D-2}) \int_{0}^{r}|\tilde{h}_1^*h_1 - \tilde{h}_1h_1^*-\tilde{h}_2^*h_2+\tilde{h}_2h^*_2|s^{D-3}~\mathrm{d}s\nonumber\\
		\leq&2Vol(\Sigma^{D-2}) \int_{0}^{r}(|\tilde{h}_1-\tilde{h}_2||h_1|+|h_1-h_2||\tilde{h}_2|)s^{D-3}~\mathrm{d}s\nonumber\\
		\leq&\frac{C\hat{x}\hat{y}r^{\frac{(2k-D)}{2}}}{(1+u)^{2k-D}(1+r+u)^{\frac{(2k-D)}{2}}}.
		\end{align}
		Using (\ref{estimate Q local}) and (\ref{Q1-Q2 local}) we represent
		\begin{align}
		\left|Q_1^2-Q_2^2\right|=|Q_1-Q_2|(|Q_1|+|Q_2|)
		\leq\frac{C\hat{x}^3\hat{y}r^{2k-D}}{(1+u)^{2(2k-D)}(1+r+u)^{2k-D}}.
		\end{align}
		We also have
		\begin{align}
		|{A_0}_1-{A_0}_2|\leq&\int_{0}^{r}\frac{|Q_1-Q_2|g_1+|g_1-g_2||Q_2|}{s^{D-2}}\mathrm{d}s\nonumber\\
		\leq& \frac{C(\hat{x}+\hat{x}^3)\hat{y}r^{D-3}} {(1+u)^{2k-3}(1+r+u)^{D-3}}.
		\end{align}
		
		We use the estimate (\ref{h-htilde local}) to show
		\begin{align}\label{g1-g2 local}
		|g_1-g_2|\leq\frac{8\pi}{(D-2)} \int_{r}^{\infty}\frac{1}{s}\left||h_1-\tilde{h}_1|^2-|h_2-\tilde{h}_2|^2\right|\mathrm{d}s\leq \frac{C\hat{x}\hat{y}}{(1+u)^{2k-D}(1+r+u)^{D-2}}.
		\end{align}
		Then, using the mean value formula we get
		\begin{align}
		|\bar{g}_1-\bar{g}_2|\leq\frac{1}{r}\int_{0}^{r}|g_1-g_2|\mathrm{d}s\leq\frac{C\hat{x}\hat{y}}{(1+u)^{2k-3}(1+r+u)}.
		\end{align}
		Therefore we obtain
		\begin{align}
		\left|\tilde{g}_1-\tilde{g}_2\right|\leq& \frac{\hat{R}(\hat{\sigma})}{(D-2)}\left|\bar{g}_1-\bar{g}_2\right|+\frac{(D-4)}{(D-2)}\frac{\hat{R}(\hat{\sigma})}{r^{D-3}}\int_{0}^{r}\left|\bar{g}_1-\bar{g}_2\right|s^{D-4}\mathrm{d}s\nonumber\\
		&+ \frac{16\pi}{(D-2)}\frac{1}{r^{D-3}}\int_{0}^{r} g_1s^{D-2}\left|V(|\tilde{h}_1|)-V(|\tilde{h}_2|)\right|\mathrm{d}s\nonumber\\
		&+\frac{16\pi}{(D-2)}\frac{1}{r^{D-3}}\int_{0}^{r}s^{D-2}|V(|\tilde{h}_1|)||g_1-g_2|\mathrm{d}s\nonumber\\
		&+\frac{2}{(D-2)Vol^2(\Sigma^{D-2})}\frac{1}{r^{D-3}}\int_{0}^{r}\frac{|Q_1^2-Q_2^2|}{s^{D-2}}g\mathrm{d}s\nonumber\\
		&+\frac{2}{(D-2)Vol^2(\Sigma^{D-2})}\frac{1}{r^{D-3}}\int_{0}^{r}\frac{|g_1-g_2|}{s^{D-2}}|Q_1|^2\mathrm{d}s\nonumber\\
		\leq&\frac{C(\hat{x}+\hat{x}^3+\hat{x}^5+\hat{x}^p+\hat{x}^{p+2})\hat{y}}{(1+u)^{2k-3}(1+r+u)}.
		\end{align}
		
		By (\ref{h1-h2 local}), (\ref{h-htilde local}), (\ref{V1-V2 local}), (\ref{Q1-Q2 local}) together with (\ref{Fcurl local}) and (\ref{Gcurl local}) we obtain
		\begin{equation}\label{varphi tilde local}
		|\tilde{\varphi}|\leq\frac{C(\hat{\alpha}(\hat{x})+\hat{\beta}(\hat{x})+\hat{\gamma}(\hat{x})+\hat{\sigma}(\hat{x})+\hat{\eta}(\hat{x})+\hat{\rho}(\hat{x})+\hat{\omega}(\hat{x})+\hat{\lambda}(\hat{x}))\hat{y}}{(1+u)^{\frac{(2k-D)k}{2}}(1+r+u)^{\frac{D}{2}}}\;,
		\end{equation}
		where
		\begin{align}
		\hat{\alpha}(\hat{x}):=&C(\hat{x}+\hat{x}^3+\hat{x}^5+\hat{x}^p+\hat{x}^{p+2})(\hat{d}+\hat{x}^3+\hat{x}^5+\hat{x}^7+\hat{x}^p+\hat{x}^{p+2}+\hat{x}^{p+4})\nonumber\\
		&\times(1+\hat{x}^2+\hat{x}^4+\hat{x}^6+\hat{x}^{p+1}+\hat{x}^{p+3})\delta\exp\left[C(\hat{x}^2+\hat{x}^4+\hat{x}^{p+1})\delta\right],\\
		\hat{\beta}(\hat{x}):=&C(\hat{x}+\hat{x}^3+\hat{x}^5+\hat{x}^p+\hat{x}^{p+2})(\hat{d}+\hat{x}^3+\hat{x}^5+\hat{x}^p+\hat{x}^{p+2})\delta\nonumber\\
		&\times\exp\left[C(\hat{x}^2+\hat{x}^4+\hat{x}^{p+1})\right],\\
		\hat{\gamma}(\hat{x}):=&C \hat{x}^{p+2}(\hat{d}+\hat{x}^3+\hat{x}^5+\hat{x}^p+\hat{x}^{p+2})\delta\exp\left[C(\hat{x}^2+\hat{x}^4+\hat{x}^{p+1})\delta\right],\\
		\hat{\sigma}(\hat{x}):=&C\hat{x}^{p}(\hat{d}+\hat{x}^3+\hat{x}^5+\hat{x}^p+\hat{x}^{p+2})\delta\exp\left[C(\hat{x}^2+\hat{x}^4+\hat{x}^{p+1})\delta\right],\\
		\hat{\eta}(\hat{x}):=&C\hat{x}^3(\hat{d}+\hat{x}^3+\hat{x}^5+\hat{x}^p+\hat{x}^{p+2})\delta\exp\left[C(\hat{x}^2+\hat{x}^4+\hat{x}^{p+1})\delta\right],\\
		\hat{\rho}(\hat{x}):=&C\hat{x}^5(\hat{d}+\hat{x}^3+\hat{x}^5+\hat{x}^p+\hat{x}^{p+2})\delta\exp\left[C(\hat{x}^2+\hat{x}^4+\hat{x}^{p+1})\delta\right],\\
		\hat{\omega}(\hat{x}):=&C(\hat{x}+\hat{x}^3)(\hat{d}+\hat{x}^3+\hat{x}^5+\hat{x}^p+\hat{x}^{p+2})\delta\exp\left[C(\hat{x}^2+\hat{x}^4+\hat{x}^{p+1})\delta\right],\\
		\hat{\lambda}(\hat{x}):=&C(\hat{x}^2+\hat{x}^4+\hat{x}^6+\hat{x}^{p-1}+\hat{x}^{p+1}+\hat{x}^{p+3}).
		\end{align}
		
		We write the estimate for the exponential term of (\ref{Theta}) from $u\in[0,u_0]$ up to $u'=\delta>u_0$ such that
		\begin{align}
		&\int_{0}^{\delta} \left[\frac{1}{2r}\left|\frac{\hat{R}(\hat{\sigma})}{(D-2)}g_1-\frac{(D-2)}{2}\tilde{g}_1\right|+\frac{8\pi r g_2|V(|\tilde{h}_1|)|}{(D-2)}+\left|\frac{g_1 Q_1^2}{(D-2)r^3}\right|+\left|i{A_0}_1\right|\right]\mathrm{d}u'\nonumber\\
		&\leq C(\hat{x}^2+\hat{x}^4+\hat{x}^{p+1})\delta.
		\end{align}
		Therefore, the estimate for (\ref{Theta}) yields
		\begin{align}
		\left|\Theta(\delta,r)\right|\leq \hat{C}_7\frac{\hat{y} \hat{M}(\hat{x})\delta\exp[\hat{C}_8(\hat{x}^2+\hat{x}^4+\hat{x}^{p+1})\delta]}{(1+r)^{k-1}},
		\end{align}
		for any $\hat{C}_7, \hat{C}_8 >0$ depends on $k$ and dimensions $D$, and we denote
		\begin{align}
		\hat{M}(\hat{x}) := \hat{\alpha}(\hat{x})+\hat{\beta}(\hat{x})+\hat{\gamma}(\hat{x})+\hat{\sigma}(\hat{x})+\hat{\eta}(\hat{x})+\hat{\rho}(\hat{x})+\hat{\omega}(\hat{x})+\hat{\lambda}(\hat{x}).
		\end{align}
		
		Finally, we represent
		\begin{equation}
		\|\Theta\|_{\hat{Y}} \leq\mathcal{L}_2\hat{y}
		\end{equation}
		with
		\begin{equation}
		\mathcal{L}_2:=\hat{C}_7M(\hat{x})\delta\exp[\hat{C}_8(\hat{x}^2+\hat{x}^4+\hat{x}^{p+1})\delta].\:
		\end{equation}
		We have $\mathcal{L}_2(0)=0$ and $\mathcal{L}_2'(0)>0$. 
		Furthermore, $\mathcal{L}_2$ is monotonically increasing on $\hat{x}_2\in \mathbb{R}^+$. There exists $\hat{x}_3\in \mathbb{R}^+$ such that $\mathcal{L}_2(x)<1$ for all $\hat{x}$ in $(0,\hat{x}_3]$. Then, we conclude that the mapping $h\mapsto\mathcal{F}(h)$ contracts in $\hat{Y}$ for $\|h\|_{\hat{X}}\leq \hat{x}_3$. This is the end of the proof.
	\end{proof}
	
	\begin{proof}[Proof of Theorem \ref{Theorem Local}]
		Since $\hat{Y}$ containing $\hat{X}$, and $h\mapsto\mathcal{F}(h)$ contracts in $\hat{Y}$, then $h\mapsto\mathcal{F}(h)$ also contracts in $\hat{X}$. Therefore, there exists a unique fixed point $h\in \hat{X}$ such that $\mathcal{F}(h)=h$. In Lemma \ref{Lemma local 1} and Lemma \ref{Lemma local 2} we have proven that $h$ and $\frac{\partial h}{\partial r}$ are bounded in $[0,u_0]\times [0,\infty)$. Now, we shall prove that $\left|\frac{\partial h}{\partial u}\right|$ also bounded in $[0,u_0]\times [0,\infty)$.
		
		From (\ref{Dh higher}) we obtain
		\begin{align}\label{estimate dhdu local}
		\left|\frac{\partial h}{\partial u}\right|\leq& \frac{\tilde{g}}{2}\left|\frac{\partial h}{\partial r}\right| + \frac{1}{2r}\left|\frac{\hat{R}(\hat{\sigma})}{(D-2)}g-\frac{(D-2)}{2}\tilde{g}\right|\left|h\right|+\frac{(D-2)}{4r}\left|\frac{\hat{R}(\hat{\sigma})}{(D-2)}g-\frac{(D-2)}{2}\tilde{g}\right|\left|\tilde{h}\right|\nonumber\\
		&+\frac{8\pi g r}{(D-2)}|V(|\tilde{h}|)||h|+4\pi g r|V(|\tilde{h}|)||\tilde{h}|+\left|\frac{gQ^2}{(D-2)r^3}\right||h|+\left|\frac{gQ^2}{2r^3}\right||\tilde{h}|\nonumber\\
		&+\left|\frac{gQ}{2r}\right||\tilde{h}|+|A_0||h|+\frac{gr}{2}\left|\frac{\partial V(|\tilde{h}|)}{\partial \tilde{h}^*}\right|\nonumber\\
		\leq&\frac{C(\hat{\xi}(\hat{x})+\hat{x}^3+\hat{x}^5+\hat{x}^p+\hat{x}^{p+2})}{(1+u)^{\frac{(2k-D)k}{2}}(1+r+u)^{\frac{(D-2)k}{2}-1}}.
		\end{align}
		where
		\begin{align}
		\hat{\xi}(\hat{x}):=&C(\hat{d}+\hat{x}^3+\hat{x}^5+\hat{x}^7+\hat{x}^p+\hat{x}^{p+2}+\hat{x}^{p+4})\nonumber\\
		&\times(1+\hat{x}^2+\hat{x}^4+\hat{x}^6+\hat{x}^{p+1}+\hat{x}^{p+3})\delta\exp\left[C(\hat{x}^2+\hat{x}^4+\hat{x}^{p+1})\delta\right].
		\end{align}
		Estimate (\ref{estimate dhdu local}) ensures that $\left|\frac{\partial h}{\partial u}\right|$ also bounded in $[0,u_0]\times [0,\infty)$. Finally, the proof is finished.  
	\end{proof}
	
	\section{Global existence}\label{Sec4}
	This section is devoted to prove the existence of global classical solution to (\ref{Dh higher}). Global existence describes that the solution of the system (\ref{Dh higher}) exists for $u_0\rightarrow\infty$.
	
	For $D\geq4$, let $k>\frac{D}{2}$. Given the space of function
	\begin{align}
	X :=& \{h\in C^1([0,\infty) \times [0,\infty) )\; | \;\|h\|_X < \infty\}\:, \label{space X}\\
	X_0 :=& \{h\in C^1([0,\infty) )\; | \;\|h\|_{X_0} < \infty\}\:,
	\end{align}
	equipped with the norm
	\begin{align}
	\|h\|_X :=&\sup_{u\geq 0} \sup_{r\geq 0}\left\{ (1+r+u)^{k-1}|h(u,r)| + (1+r+u)^{k}\left|\frac{\partial h}{\partial r}(u,r)\right|  \right\}\:, \label{hx}\\
	\|h\|_{X_0} :=& \sup_{r\geq 0}\left\{ (1+r)^{k-1}|h(r)| + (1+r)^{k}\left|\frac{\partial h}{\partial r}(r)\right|  \right\}.\label{hx0}
	\end{align}
	We also define the space of function
	\begin{align}
	Y := \{h\in C^1([0,\infty) \times [0,\infty) )\; | \; h(0,r) = h_0(r), \|h\|_Y < \infty\}\label{space Y}\:,
	\end{align}
	in which $Y$ containing $X$, equipped by the norm
	\begin{align}
	\|h\|_Y := \sup_{u\geq 0} \sup_{r\geq 0}\left\{ (1+r+u)^{k-1}|h(u,r)|\right\}\:\label{norm Y}.
	\end{align}
	Here, the space function $X,$ $X_0,$, and $Y$ are the Banach spaces. To simplify, let us denote $\|h\|_X:=x$, $\|h_0\|_{X_0}:=d$, and $\|h_1-h_2\|_Y:=y$.
	
	We state the existence of global solutions of (\ref{Dh higher}) as follows.
	\begin{theorem} \label{Theorem Global}
		Let $D\geq 4$. Given a positive constant $(D-2)$-spatial Ricci scalar of compact manifold $\hat{R}(\hat{\sigma})$.
		Suppose that $X$ and $Y$ are the function spaces defined by (\ref{space X}) and (\ref{space Y}) respectively. For a given initial data $h(0,r)\in C^1[0,\infty)$ such that $h(0,r)=O(r^{-{(k-1)}})$ and $\frac{\partial h}{\partial r}(0,r)=O(r^{-k})$ as $r\rightarrow \infty$, there exists a global classical solution of equation (\ref{Dh higher}) satisfies
		\begin{align}
		h(u,r)\in C^1 ([0,\infty)\times [0,\infty)),
		\end{align}
		for $p\in[k,\infty)$ and $k>\frac{D}{2}$, where $p,k\in\mathbb{R}^+$.
	\end{theorem}
	
	\subsection{Decay estimates}
	First, we study the decay estimates of the solutions by the following Lemma.
	\begin{lemma} \label{Lemma Decay}
		Let $\|h\|_X:=x$ and $\|h(0,.)\|_{X_0}:=d$. For $D\geq4$, setting $p\in[k,\infty)$ and $k>\frac{D}{2}$, where $p,k\in\mathbb{R}^+$. We introduce a positive constant $(D-2)$-spatial Ricci scalar of compact manifold $\hat{R}(\hat{\sigma})$. Given the initial data $h(0,r)\in C^1[0,\infty)$ such that $h(0,r)=O(r^{-{(k-1)}})$ and $\frac{\partial h}{\partial r}(0,r)=O(r^{-k})$ as $r\rightarrow \infty$.  Then, the solution of (\ref{Dh higher}) fulfilles the following the decay estimates:
		\begin{align} \label{decay1}
		|h(u,r)|\leq \frac{C(d+x^3+x^5+x^p+x^{p+2})\exp\left[C(x^2+x^4+x^{p+1})\right]}{(1+r+u)^{k-1}},
		\end{align}
		and
		\begin{align} \label{decay2}
		\left|\frac{\partial h}{\partial r}(u,r) \right|\leq& \frac{C(d+x^3+x^5+x^7+x^p+x^{p+2}+x^{p+4})(1+x^2+x^4+x^6+x^{p+1}+x^{p+3})}{(1+r+u)^k}\nonumber\\
		&\times\exp\left[C(x^2+x^4+x^{p+1})\right]
		\end{align}
		for any $C:=C(k,D)>0$ depends on $k$ and dimensions $D$.
	\end{lemma}
	\begin{remark}
		In the previous work, Chae \cite{Chae1} assume that the initial data satisfies $h(0,r)=O(r^{-{2}})$ and $\frac{\partial h}{\partial r}(0,r)=O(r^{-3})$. Therefore the decay order for solutions of the EMH system in four dimensions satisfies $\left|h(u,r)\right|+r\left|\frac{\partial h(u,r)}{\partial r}\right|=O(r^{-2})$. Motivated by \cite{Prisma,Prisma2}, in this paper we obtain more general order of decay estimates for solutions of the EMH system in higher dimensions $D\geq 4$ in the following $\left|h(u,r)\right|+r\left|\frac{\partial h(u,r)}{\partial r}\right|=O(r^{-{(k-1)}})$, with $k> \frac{D}{2}$.
	\end{remark}
	
	\begin{proof}
		Using the definition of norm (\ref{hx}), we calculate
		\begin{align}\label{estimasi htilde}
		|\tilde{h}|\leq\frac{2x}{(2k-D)}\frac{1}{(1+u)^{\frac{(2k-D)}{2}}(1+r+u)^{\frac{(D-2)}{2}}}.
		\end{align}
		Thus we have
		\begin{align}\label{h-h}
		\left|h-\frac{(D-2)}{2}\tilde{h}\right|\leq  \frac{Cx}{(1+u)^{\frac{(2k-D)}{2}}(1+r+u)^{\frac{(D-2)}{2}}}.
		\end{align}
		The above estimate yields
		\begin{align}\label{g-gbar k}
		|(g-\bar{g})(u,r)|\leq\frac{8\pi x^2}{(D-3)(D-2)^2}\frac{1}{(1+u)^{2k-3}(1+r+u)}.
		\end{align}
		Using triangle inequality we obtain
		\begin{align}\label{estimate gbar}
		|\bar{g}|\geq\frac{8\pi x^2}{(D-3)(D-2)^2}\frac{1}{(1+u)^{2k-3}(1+r+u)},
		\end{align}
		Setting $p\in[k,\infty)$, such that
		\begin{align}
		\frac{1}{r^{D-3}}\int_{0}^{r} gs^{D-2}|V(|\tilde{h}|)|\mathrm{d}s\leq \frac{Cx^{p+1}}{(1+u)^{k^2-D}(1+r+u)^{D-3}}.
		\end{align}
		From the definition of local charge (\ref{Qh}) we calcultae
		\begin{align}\label{estimate Q}
		|Q(u,r)|\leq\frac{8Vol (\Sigma^{D-2})}{(2k-D)^2}\frac{x^2 r^{\frac{(2k-D)}{2}}}{(1+u)^{2k-D}(1+r+u)^{\frac{(2k-D)}{2}}}.
		\end{align}
		As a consequences, we have
		\begin{align}\label{estimate integral Q}
		\frac{1}{Vol^2(\Sigma^{D-2})}\frac{1}{r^{D-3}}\int_{0}^{r}\left|\frac{gQ^2}{s^{D-2}}\right|\mathrm{d}s\leq\frac{2^6}{(D-3)(2k-D)^4}\frac{x^4}{(1+u)^{4k-D-3}(1+r+u)^{D-3}},
		\end{align}
		and
		\begin{align}\label{estimate A0}
		|iA_0|\leq \frac{8x^2}{(2k-D)^2(D-3)}\frac{r^{D-3}}{(1+u)^{2k-3}(1+r+u)^{D-3}}.
		\end{align}
		Substituting (\ref{g-gbar k}), (\ref{estimate gbar}), (\ref{estimate integral Q}), and (\ref{estimate A0}) into (\ref{g tilde}) yields
		\begin{align}\label{g - gtilde}
		\left|\frac{\hat{R}(\hat{\sigma})}{(D-2)}g-\frac{(D-2)}{2}\tilde{g}\right|\leq\frac{C(x^2+x^4+x^{p+1})}{(1+u)^{2k-3}(1+r+u)}.
		\end{align}
		By (\ref{potential}), (\ref{htilde}), (\ref{estimate Q}), and (\ref{g - gtilde}), we obtain (see Appendix \ref{Appendix1})
		\begin{align}\label{estimate f}
		|f|\leq \frac{C(x^3+x^5+x^p+x^{p+2})}{(1+u)^\frac{(2k-D)k}{2}(1+r+u)^{\frac{D}{2}}}.
		\end{align}	
		
		From (\ref{g tilde}) we obtain
		\begin{align}
		|\tilde{g}(u,0)|\geq&\frac{2^4\pi x^2\hat{R}(\hat{\sigma})}{(D-3)(D-2)^3}+\frac{2^{p+7}\pi K_0x^{p+1}}{(D-2)\left[(D-2)k-D\right](2k-D)^{k+1}}\nonumber\\
		&+\frac{2^6x^4}{(D-3)(2k-D)^4}.
		\end{align}
		Let $x_1$ be the positive solution of
		\begin{align}\label{x1}
		\frac{2^4\pi x^2\hat{R}(\hat{\sigma})}{(D-3)(D-2)^3}+\frac{2^{p+7}\pi K_0x^{p+1}}{(D-2)\left[(D-2)k-D\right](2k-D)^{k+1}}+\frac{2^6x^4}{(D-3)(2k-D)^4}=0.
		\end{align}
		Now, let us introduce a new function
		\begin{align}\label{k higher}
		\kappa:=\kappa(x)=&\frac{2^4\pi x^2\hat{R}(\hat{\sigma})}{(D-3)(D-2)^3}+\frac{2^{p+7}\pi K_0x^{p+1}}{(D-2)\left[(D-2)k-D\right](2k-D)^{k+1}}\nonumber\\
		&+\frac{2^6x^4}{(D-3)(2k-D)^4},
		\end{align}
		for all $x\in[0,x_1)$. From the characteristics equation we have
		\begin{align}\label{ru higher}
		r(u)=r_1 + \frac{1}{2}\int_{u}^{u_1}\tilde{g}(u',r(u'))\mathrm{d}u'\geq r_1 +\frac{1}{2}\kappa(u_1-u),
		\end{align}
		and
		\begin{align}\label{ru2 higher}
		1+r(u)+u\geq 1 + \frac{u}{2} + r_1 + \frac{\kappa}{2}(u_1 - u)\geq \frac{\kappa}{2}(1+r_1+u_1).
		\end{align}
		
		We use (\ref{ru2 higher}) and the definition of norm (\ref{hx0}) to produce
		\begin{align}\label{estimate h0}
		|h(0,r_0)|\leq \frac{d}{(1+r_0)^{k-1}}\leq \frac{d}{(1+r_1+\frac{1}{2}\kappa u_1)^{k-1}}\leq \frac{2^{k-1}d}{\kappa^{k-1}(1+r_1+u_1)^{k-1}}.
		\end{align}
		
		Let us define the integral formula
		\begin{align}\label{integral}
		\int_{0}^{u_1}\left[\frac{r^s}{(1+u)^t(1+r+u)^q}\right]_\chi\mathrm{d}u\leq& \int_{0}^{u_1}\left[\frac{1}{(1+u)^t(1+r+u)^{q-s}}\right]_\chi\mathrm{d}u\nonumber\\
		\leq& \frac{1}{\kappa^m(1+r_1+u_1)^m}\int_{0}^{\infty}\frac{\mathrm{d}u}{(1+u)^{q-s+t-m}}\nonumber\\
		=&\frac{2^m}{(q-s+t-m-1)\kappa^m(1+r_1+u_1)^m}, 
		\end{align}
		where $q-s+t-m>1$, with $q, s, t, m\in \mathbb{R}$. Setting $q=2$, $s=0$, $t=3$, and $m=0$.
		
		Using the integral formula (\ref{integral}), we represent the estimate for exponential term of (\ref{Fcurl}) as follows
		\begin{align}
		&\int_{0}^{u_1} \left[\frac{1}{2r}\left|\frac{\hat{R}(\hat{\sigma})}{(D-2)}g-\frac{(D-2)}{2}\tilde{g}\right|+\frac{8\pi g r}{(D-2)}|V(|\tilde{h}|)|+\frac{|Q|^2g}{(D-2)r^3}+|iA_0|\right]_\chi\mathrm{d}u\nonumber\\
		&\leq I_1 + I_2 + I_3 + I_4.
		\end{align}
		\begin{enumerate}
			\item Estimate for $I_1$\\
			Setting $q=2$, $s=0$, $t=2k-3$, and $m=0$, we obtain
			\begin{align}\label{I1 F}
			\int_{0}^{u_1}\left[\frac{1}{2r}\left|\frac{\hat{R}(\hat{\sigma})}{(D-2)}g-\frac{(D-2)}{2}\tilde{g}\right|\right]_\chi\mathrm{d}u\leq&C(x^2 + x^4 + x^{p+1})\int_{0}^{\infty}\frac{1}{(1+u)^{2k-1}}\mathrm{d}u\nonumber\\
			\leq& \frac{C(x^2 + x^4 + x^{p+1})}{2k-2}.
			\end{align}
			\item Estimate for $I_2$\\
			Setting $q=\frac{(D-2)(k+1)}{2}$, $s=1$, $t=\frac{(2k-D)(k+1)}{2}$, and $m=0$, we obtain
			\begin{align}\label{I2 F}
			\int_{0}^{u_1}\left[\frac{8\pi g r}{(D-2)}|V(|\tilde{h}|)|\right]_\chi\mathrm{d}u\leq C(x^2 + x^4 + x^{p+1})\int_{0}^{\infty}\frac{1}{(1+u)^{k^2-2}}\mathrm{d}u\leq \frac{C x^{p+1}}{k^2-3}.
			\end{align}
			\item Estimate for $I_3$\\
			Setting $q=3$, $s=0$, $t=2(2k-D)$, and $m=0$, we obtain
			\begin{align}\label{I3 F}
			\int_{0}^{u_1}\left[\frac{|Q|^2g}{(D-2)r^3}\right]_\chi\mathrm{d}u\leq Cx^4\int_{0}^{\infty}\frac{1}{(1+u)^{2(2k-D)+3}}\mathrm{d}u\leq \frac{C x^{4}}{2(2k-D+1)}.
			\end{align}
			\item Estimate for $I_4$\\
			By (\ref{estimate A0}) we have $q=D-3$, $s=D-3$, $t=2k-3$, and $m=0$. Thus we obtain
			\begin{align}\label{I4 F}
			\int_{0}^{u_1}\left[|iA_0|\right]_\chi\mathrm{d}u\leq Cx^2\int_{0}^{\infty}\frac{1}{(1+u)^{2k-3}}\mathrm{d}u\leq \frac{C x^{2}}{2k-4}.
			\end{align}
		\end{enumerate}
		By (\ref{I1 F}), (\ref{I2 F}), (\ref{I3 F}), and (\ref{I4 F}), and setting $k>\frac{D}{2}$ we obtain
		\begin{align}\label{estimate exp1}
		&\int_{0}^{u_1} \left[\frac{1}{2r}\left|\frac{\hat{R}(\hat{\sigma})}{(D-2)}g-\frac{(D-2)}{2}\tilde{g}\right|+\frac{8\pi g r}{(D-2)}|V(|\tilde{h}|)|+\frac{|Q|^2g}{(D-2)r^3}+|iA_0|\right]_\chi\mathrm{d}u\nonumber\\
		&\leq C(x^2+x^4+x^{p+1}).
		\end{align}
		Combining (\ref{estimate f}), (\ref{estimate h0}), and (\ref{estimate exp1}), we represent the estimate for (\ref{Fcurl}) as follows
		\begin{align}
		|\mathcal{F}|\leq&\frac{2^{k-1}d}{\kappa^{k-1}(1+r_1+u_1)^{k-1}}\exp\left[C(x^2+x^4+x^{p+1})\right]\nonumber\\
		&+\int_{0}^{u_1}\frac{C(x^3+x^5+x^p+x^{p+2})}{(1+u)^{\frac{(2k-D)k}{2}}(1+r+u)^{\frac{D}{2}}}\exp\left[C(x^2+x^4+x^{p+1})\right]\mathrm{d}u.
		\end{align}
		Using the integral formula (\ref{integral}) for $q=\frac{D}{2}, s=0, t=\frac{(2k-D)k}{2}, m=k-1$ we obtain
		\begin{align}\label{estimate Fcurl}
		|\mathcal{F}(u_1,r_1)|\leq\frac{C(d+x^3+x^5+x^p+x^{p+2})\exp\left[C(x^2+x^4+x^{p+1})\right]}{\kappa^{k-1}(1+r_1+u_1)^{k-1}}.
		\end{align}
		
		From (\ref{g2}) and (\ref{dgtilde}) we obtain
		\begin{align}
		\left|\frac{1}{2r}\left(\frac{\hat{R}(\hat{\sigma})}{(D-2)}\frac{\partial g}{\partial r}-\frac{(D-2)}{2}\frac{\partial\tilde{g}}{\partial r}\right)\right|\leq\frac{C(x^2+x^4+x^{p+1})}{(1+u)^{2k-D}(1+r+u)^3}.
		\end{align}
		Then, from (\ref{htilde}) we can show
		\begin{align}
		\left|\frac{\partial \tilde{h}}{\partial r}\right|\leq\frac{Cx}{(1+u)^{\frac{(2k-D)}{2}}(1+r+u)^{\frac{D}{2}}}.
		\end{align}
		On the other hand, from (\ref{estimate Q}) we get
		\begin{align}\label{dQdr}
		\left|\frac{\partial Q}{\partial r}\right|\leq\frac{Cx^2}{(1+u)^{\frac{(2k-D)}{2}}(1+r+u)^{\frac{(2k-D+2)}{2}}}.
		\end{align}
		Finally, we write the following estimate (\ref{f1}) (see Appendix \ref{Appendix2})
		\begin{align}\label{estimate f1}
		|f_1|\leq\frac{C(x^2+x^4+x^6+x^{p+1}+x^{p+3})}{(1+u)^{2k-D}(1+r+u)^{D-2}}|\mathcal{F}|+\frac{C(x^3+x^5+x^7+x^p+x^{p+2}+x^{p+4})}{(1+u)^{2k-D}(1+r+u)^\frac{(D+2)}{2}}.
		\end{align}
		
		Again, we use the definition of (\ref{norm Y}) such that
		\begin{align}\label{estimate dh0}
		\left|\frac{\partial h}{\partial r}(0,r_0)\right|\leq \frac{\|h_0\|_{X_0}}{(1+r_0)^k}\leq \frac{d}{(1+r_1+\frac{1}{2}\kappa u_1)^k}\leq \frac{2^kd}{\kappa^k(1+r_1+u_1)^k}.
		\end{align}
		As previously, using the integral formula (\ref{integral}) we obtain
		\begin{align}\label{estimate exp2}
		&\int_{0}^{u_1}\left[\frac{1}{2}\left|\frac{\partial \tilde{g}}{\partial r}\right|+\frac{1}{2r}\left|\frac{\hat{R}(\hat{\sigma})}{(D-2)}g-\frac{(D-2)}{2}\tilde{g}\right|+\frac{8\pi gr}{(D-2)} |V(|\tilde{h}|)|\right.\nonumber\\
		&\left.+\frac{g|Q|^2}{(D-2)r^3}+|iA_0|\right]_\chi\mathrm{d}u\leq C(x^2+x^4+x^{p+1}).
		\end{align}
		Combinations of (\ref{estimate f1}), (\ref{estimate dh0}), and (\ref{estimate exp2}), yields
		\begin{align}
		|\mathcal{G}|\leq&\frac{2^kd}{\kappa^k(1+r_1+u_1)^k}\exp\left[C(x^2+x^4+x^{p+1})\right]\nonumber\\
		&+\int_{0}^{u_1}\left[\frac{C(x^2+x^4+x^6+x^{p+1}+x^{p+3})}{(1+u)^{2k-D}(1+r+u)^{D-2}}|\mathcal{F}|+\frac{C(x^3+x^5+x^7+x^p+x^{p+2}+x^{p+4})}{(1+u)^{2k-D}(1+r+u)^\frac{(D+2)}{2}}\right]\nonumber\\
		&\times \exp\left[C(x^2+x^4+x^{p+1})\right]\mathrm{d}u,\nonumber\\
		\leq&\frac{2^kd}{\kappa^k(1+r_1+u_1)^k}\exp\left[C(x^2+x^4+x^{p+1})\right]+I_5+I_6.
		\end{align}
		We calculate thes estimate for $I_5$ and $I_6$ as follows
		\begin{enumerate}
			\item Estimate for $I_5$\\
			Setting $q=2, s=0, t=2k-D, m=1$ we obtain
			\begin{align}
			I_5=& \frac{C(x^2+x^4+x^6+x^{p+1}+x^{p+3})(d+x^3+x^5+x^p+x^{p+2})}{\kappa^{k}(1+r_1+u_1)^{k}}\nonumber\\
			&\times\exp\left[C(x^2+x^4+x^{p+1})\right]\int_{0}^{\infty}\frac{\mathrm{d}u}{(1+u)^{2k-D+1}}\nonumber\\
			\leq& \frac{C(x^2+x^4+x^6+x^{p+1}+x^{p+3})(d+x^3+x^5+x^p+x^{p+2})\exp\left[C(x^2+x^4+x^{p+1})\right]}{(2k-D)\kappa^{k}(1+r_1+u_1)^{k}}
			\end{align}
			\item Estimate for $I_6$\\
			Setting $q=\frac{(D+2)}{2}, s=0, t=2k-D, m=k$ we obtain
			\begin{align}
			I_6=& \frac{C(x^3+x^5+x^7+x^p+x^{p+2}+x^{p+4})\exp\left[C(x^2+x^4+x^{p+1})\right]}{\kappa^k(1+r_1+u_1)^k}\int_{0}^{\infty}\frac{\mathrm{d}u}{(1+u)^{\frac{2k-D+2}{2}}}\nonumber\\
			&\leq \frac{C(x^3+x^5+x^7+x^p+x^{p+2}+x^{p+4})\exp\left[C(x^2+x^4+x^{p+1})\right]}{(2k-D)\kappa^k(1+r_1+u_1)^k}.
			\end{align}
		\end{enumerate}
		Finally we obtain
		\begin{align}\label{estimate Gcurl}
		\left|\mathcal{G}(u_1,r_1)\right|\leq&\frac{C(d+x^3+x^5+x^7+x^p+x^{p+2}+x^{p+4})(1+x^2+x^4+x^6+x^{p+1}+x^{p+3})}{\kappa^k(1+r_1+u_1)^k}\nonumber\\
		&\times\exp\left[C(x^2+x^4+x^{p+1})\right].
		\end{align}
		This is the end of the proof.
	\end{proof}
	\subsection{Contraction mapping}	
	Now we study the contraction mapping by the following Lemma.
	
	\begin{lemma}\label{Lemma Contraction}
		Let us define $B(0,x):=\left\{f\in X| \|f\|_X\leq x \right\}$ as the close ball with radius $x$ in $X$. For a suitable $x$ and $d:=d(x)$, then $\mathcal{F}(.)$ satisfies the contraction mapping arguments:
		\begin{enumerate}
			\item $\mathcal{F}: B(0,x)\rightarrow B(0,x)$, and
			\item Suppose that equation (\ref{Dh higher}) has two solutions namely $h_1,h_2\in B(0,x)$. There exists $\tilde{\Lambda}:=\tilde{\Lambda}(x)\in[0,1)$, such that
			\begin{equation}
			\|\mathcal{F}(h_1)-\mathcal{F}(h_2)\|_Y \leq \tilde{\Lambda}\|h_1-h_2\|_Y.
			\end{equation}
			Thus, $\mathcal{F}$ contracts in $Y$.
		\end{enumerate}
	\end{lemma}
	\begin{proof}
		We start to prove the first argument. By (\ref{hx}), we have
		\begin{align}\label{Estimate F}
		\|\mathcal{F}\|_X\leq&\frac{C_1}{\kappa^k}(d+x^3+x^5+x^7+x^p+x^{p+2}+x^{p+4})\nonumber\\
		&\times(1+x^2+x^4+x^6+x^{p+1}+x^{p+3})\exp\left[C_2(x^2+x^4+x^{p+1})\right],
		\end{align}
		for any $C_1,C_2>0$ depends on $k$ and dimensions $D$.
		
		Let us define a function,
		\begin{align}
		\tilde{\Lambda}_1(x):=&\frac{x\kappa^k\exp\left[-C_2(x^2+x^4+x^{p+1})\right]}{C_1(1+x^2+x^4+x^6+x^{p+1}+x^{p+3})}\nonumber\\
		&-\left(x^3+x^5+x^7+x^p+x^{p+2}+x^{p+4}\right).
		\end{align}
		We have $\tilde{\Lambda}_1(0)=0$, $\tilde{\Lambda}_1'(0)>0$, and $\tilde{\Lambda}_1(x)\rightarrow-\infty$ as $x\rightarrow\infty$. There exists $x_0\in(0,x_1)$ such that $\tilde{\Lambda}_1(x)$ is monotonically increasing on $[0,x_0]$. Therefore, if $d<\tilde{\Lambda}_1(x)$, we have $\|\mathcal{F}\|_X\leq x$ and $\mathcal{F}:B(0,x)\rightarrow B(0,x)$ for every $x\in(0,x_0)$. This proves the first argument. 
		
		Now, we prove the second argument. We firstly calculate 
		\begin{align}\label{h1-h2}
		|\tilde{h}_1 - \tilde{h}_2|\leq\frac{2y}{(2k-D)}\frac{1}{(1+u)^{\frac{(2k-D)}{2}}(1+r+u)^{\frac{(D-2)}{2}}}.
		\end{align}
		Next, we calculate
		\begin{align}\label{V1-V2}
		\left|V(|\tilde{h}_1|)-V(|\tilde{h}_2|)\right|\leq\frac{2^{2p+1}x^py}{(2k-D)^{k+1}(1+u)^{\frac{(2k-D)(k+1)}{2}}(1+r+u)^{\frac{(D-2)(k+1)}{2}}},
		\end{align}
		and also
		\begin{align}
		\left|\frac{\partial V(|\tilde{h}_1|)}{\partial \tilde{h}_1^*}-\frac{\partial V(|\tilde{h}_2|)}{\partial \tilde{h}_2^*}\right|
		\leq\frac{C x^{p-1}y}{(1+u)^{\frac{(2k-D)k}{2}}(1+r+u)^{\frac{(D-2)k}{2}}}.
		\end{align}
		
		From the estimate (\ref{estimasi htilde}) and (\ref{h1-h2}) we obtain
		\begin{align}\label{Q1-Q2}
		|Q_1-Q_2|\leq\frac{Cxyr^{\frac{(2k-D)}{2}}}{(1+u)^{2k-D}(1+r+u)^{\frac{(2k-D)}{2}}}.
		\end{align}
		Thus, combinations of (\ref{estimate Q}) and (\ref{Q1-Q2}) produces
		\begin{align}
		\left|Q_1^2-Q_2^2\right|=|Q_1-Q_2|(|Q_1|+|Q_2|)
		\leq\frac{Cx^3yr^{2k-D}}{(1+u)^{2(2k-D)}(1+r+u)^{2k-D}},
		\end{align}
		and
		\begin{align}
		|{A_0}_1-{A_0}_2|\leq\frac{C(x+x^3)yr^{D-3}}{(1+u)^{2k-3}(1+r+u)^{D-3}}.
		\end{align}
		
		We can verify
		\begin{align}\label{g1-g2}
		|g_1-g_2|\leq\frac{Cxy}{(1+u)^{2k-D}(1+r+u)^{D-2}}.
		\end{align}
		From (\ref{g1-g2}), we have
		\begin{align}
		|\bar{g}_1-\bar{g}_2|\leq\frac{1}{r}\int_{0}^{r}|g_1-g_2|\mathrm{d}s\leq\frac{Cxy}{(1+u)^{2k-3}(1+r+u)},
		\end{align}
		and
		\begin{align}
		|\tilde{g}_1-\tilde{g}_2|\leq\frac{C(x+x^3+x^5+x^p+x^{p+2})y}{(1+u)^{2k-3}(1+r+u)}.
		\end{align}
		
		Using all the above estimates and combining with (\ref{estimate Fcurl}) and (\ref{estimate Gcurl}), we write down the estimate for (\ref{phi tilde}) as follows (see Appendix \ref{Appendix3})
		\begin{equation}\label{varphi tilde}
		|\tilde{\varphi}|\leq\frac{C(\alpha(x)+\beta(x)+\gamma(x)+\sigma(x)+\eta(x)+\rho(x)+\omega(x)+\lambda(x))y}{(1+u)^{\frac{(2k-D)k}{2}}(1+r+u)^{\frac{D}{2}}}\;,
		\end{equation}
		where
		\begin{align}
		\alpha(x):=&C(x+x^3+x^5+x^p+x^{p+2})(d+x^3+x^5+x^7+x^p+x^{p+2}+x^{p+4})\nonumber\\
		&\times(1+x^2+x^4+x^6+x^{p+1}+x^{p+3})\exp\left[C(x^2+x^4+x^{p+1})\right],\label{alpha}\\
		\beta(x):=&C(x+x^3+x^5+x^p+x^{p+2})(d+x^3+x^5+x^p+x^{p+2})\nonumber\\
		&\times\exp\left[C(x^2+x^4+x^{p+1})\right],\label{beta}\\
		\gamma(x):=&C x^{p+2}(d+x^3+x^5+x^p+x^{p+2})\exp\left[C(x^2+x^4+x^{p+1})\right],\label{gamma}\\
		\sigma(x):=&Cx^{p}(d+x^3+x^5+x^p+x^{p+2})\exp\left[C(x^2+x^4+x^{p+1})\right],\label{sigma}\\
		\eta(x):=&Cx^3(d+x^3+x^5+x^p+x^{p+2})\exp\left[C(x^2+x^4+x^{p+1})\right],\label{eta}\\
		\rho(x):=&Cx^5(d+x^3+x^5+x^p+x^{p+2})\exp\left[C(x^2+x^4+x^{p+1})\right],\label{rho}\\
		\omega(x):=&C(x+x^3)(d+x^3+x^5+x^p+x^{p+2})\exp\left[C(x^2+x^4+x^{p+1})\right],\label{omega}\\
		\lambda(x):=&C(x^2+x^4+x^6+x^{p-1}+x^{p+1}+x^{p+3}).\label{lambda}
		\end{align}
		
		We use the similar method as in (\ref{estimate exp1}) to produce
		\begin{align}\label{estimate exp3}
		&\int_{u}^{u_1} \left[\frac{1}{2r}\left|\frac{\hat{R}(\hat{\sigma})}{(D-2)}g_1-\frac{(D-2)}{2}\tilde{g}_1\right|+\frac{8\pi r g_2}{(D-2)} \left|V(|\tilde{h}_1|)\right|+\left|\frac{g_1Q_2^2}{(D-2)r^3}\right|+\left|i{A_0}_1\right|\right]_{\chi_1}\mathrm{d}u'\nonumber\\
		&\leq C(x^2+x^4+x^{p+1}).
		\end{align}
		We assume
		\begin{align}
		\max\left\{\|\tilde{h}_1\|_X,\|\tilde{h}_2\|_X\right\}<x.
		\end{align}
		By (\ref{varphi tilde}) and (\ref{estimate exp3}), we represent the estimate for (\ref{Theta}) as follows
		\begin{align}
		|\Theta(u_1,r_1)|\leq\int_{0}^{u_1}\frac{CM(x)y}{(1+u)^{\frac{(2k-D)k}{2}}(1+r+u)^{\frac{D}{2}}}\exp\left[C(x^2+x^4+x^{p+1})\right]\mathrm{d}u,
		\end{align}
		where
		\begin{align}
		M(x):=\alpha(x)+\beta(x)+\gamma(x)+\sigma(x)+\eta(x)+\rho(x)+\omega(x)+\lambda(x).
		\end{align}
		Again, using the integral formula (\ref{integral}) for $q=\frac{D}{2}, s=0, t=\frac{(2k-D)k}{2}, m=k-1$ we obtain
		\begin{align}
		&\frac{CM(x)y \exp\left[C(x^2+x^4+x^{p+1})\right]}{\kappa^{k-1}(1+r_1+u_1)^{k-1}}\int_{0}^{\infty}\frac{1}{(1+u)^{\frac{(D-2k)(1-k)+2}{2}}}\mathrm{d}u\nonumber\\
		&\leq \frac{CM(x)y \exp\left[C(x^2+x^4+x^{p+1})\right]}{\kappa^{k-1}(1+r_1+u_1)^{k-1}(k-1)(2k-D)}.
		\end{align}
		Hence, we have
		\begin{align}\label{estimate Theta}
		|\Theta(u_1,r_1)|\leq\frac{C_3M(x)y\exp\left[C_4(x^2+x^4+x^{p+1})\right]}{\kappa^{k-1}(1+r_1+u_1)^{k-1}},
		\end{align}
		for any $C_3,C_4>0$ depends on $k$ and dimensions $D$.
		
		Finally, from the definition of norm (\ref{norm Y}) we have
		\begin{equation}
		\|\Theta\|_Y \leq\tilde{\Lambda}_2y
		\end{equation}
		with
		\begin{equation}
		\tilde{\Lambda}_2:=\frac{C_3}{\kappa^{k-1}}M(x)\exp\left[C_4(x^2+x^4+x^{p+1})\right].\:
		\end{equation}
		We have $\tilde{\Lambda}_2(0)=0$ and $\tilde{\Lambda}'_2(0)>0$. Also, we conclude that $\tilde{\Lambda}_2$ is monotonically increasing function on $\tilde{x}_0\in \mathbb{R}^+$. There exists $x_2\in \mathbb{R}^+$ such that $\tilde{\Lambda}_2(x)<1$ for all $x$ in $(0,x_2]$. 
		Therefore, the mapping $h\mapsto\mathcal{F}(h)$ is contraction in $Y$ for $\|h\|_X\leq x_2$. The is the end of the proof.
	\end{proof}
	
	\subsection{Proof of Theorem \ref{Theorem Global}}
	We prove Theorem \ref{Theorem Global} using the results we obtained in Lemma \ref{Lemma Decay} and Lemma \ref{Lemma Contraction}. In Lemma \ref{Lemma Decay}, we show that $h$ and $\frac{\partial h}{\partial r}$ are bounded in $[0,\infty)\times [0,\infty)$. Furthermore, in Lemma \ref{Lemma Contraction} we prove the arguments (1) and (2) which are the main arguments of the Banach fixed theorem. These provide the existence of unique fixed point $\mathcal{F}(h)=h$ which is the solution of the equation (\ref{Dh higher}). Now, we will show that $h$ is a classical solution. Therefore, we shall prove that $\frac{\partial h}{\partial u}$ is also bounded in $[0,\infty)\times [0,\infty)$.
	
	Given that $Y$ containing $X$, since $h\mapsto\mathcal{F}(h)$ is contraction mapping in $Y$, it is straightforward to show that $h\mapsto\mathcal{F}(h)$ is also contraction mapping in $X$. As a consequences, there also exists a unique fixed point $h\in X$ such that $\mathcal{F}(h)=h$.
	
	By (\ref{Dh higher}), (\ref{estimasi htilde}), (\ref{estimate Q}), (\ref{estimate Fcurl}), and (\ref{estimate Gcurl}), we obtain
	\begin{align}\label{estimate dhdu}
	\left|\frac{\partial h}{\partial u}\right|\leq& \frac{\tilde{g}}{2}\left|\frac{\partial h}{\partial r}\right| + \frac{1}{2r}\left|\frac{\hat{R}(\hat{\sigma})}{(D-2)}g-\frac{(D-2)}{2}\tilde{g}\right|\left|h\right|+\frac{(D-2)}{4r}\left|\frac{\hat{R}(\hat{\sigma})}{(D-2)}g-\frac{(D-2)}{2}\tilde{g}\right|\left|\tilde{h}\right|\nonumber\\
	&+\frac{8\pi g r}{(D-2)}|V(|\tilde{h}|)||h|+4\pi g r|V(|\tilde{h}|)||\tilde{h}|+\left|\frac{gQ^2}{(D-2)r^3}\right||h|+\left|\frac{gQ^2}{2r^3}\right||\tilde{h}|\nonumber\\
	&+\left|\frac{gQ}{2r}\right||\tilde{h}|+|A_0||h|+\frac{gr}{2}\left|\frac{\partial V(|\tilde{h}|)}{\partial \tilde{h}^*}\right|\nonumber\\
	\leq&\frac{C(\xi(x)+x^3+x^5+x^p++x^{p+2})}{(1+u)^{\frac{(2k-D)k}{2}}(1+r+u)^{\frac{(D-2)k}{2}-1}},
	\end{align}
	where
	\begin{align}
	\xi(x):=&C(d+x^3+x^5+x^7+x^p+x^{p+2}+x^{p+4})\nonumber\\
	&\times(1+x^2+x^4+x^6+x^{p+1}+x^{p+3})\exp\left[C(x^2+x^4+x^{p+1})\right].
	\end{align}
	Thus, we conclude that $\frac{\partial h}{\partial u}$ is bounded  in $[0,\infty)\times [0,\infty)$. This is the end of the proof.
	
	\section{Completeness properties of the spacetime}
	\label{Sec5}
	Motivated by \cite{Chris1,Chae1,Prisma2}, in the following section we study the completeness properties of the higher dimensional spacetime ($D\geq 4$). We firstly introduce a local mass function for $D\geq 4$ as follows:
	\begin{align}\label{M}
	M(u,r) :=\frac{r^{D-3}}{2}\left[\frac{2\hat{R}(\hat{\sigma})}{(D-2)^2}-\frac{\tilde{g}}{g}+\frac{2}{(D-3)(D-2)Vol^2(\Sigma^{D-2})}\frac{Q^2}{r^{2(D-3)}}\right].
	\end{align}
	By taking $D=4$ and $Q=0$, equation (\ref{M}) reduce to equation (4.3) in \cite{Chris1}. Furthermore, in the view of	(\ref{g2}) we shall note that $g$ is monotonically nondecreasing function of $r$ at each $u$, and also from (\ref{g - gtilde}) we obtain $\tilde{g}\leq \frac{2\hat{R}(\hat{\sigma})}{(D-2)^2} g$. This implies that $M(u,r)\geq 0$.	For each $u\in [0,\infty)$, setting
	\begin{align}
	\sup_{0<r\infty} M(u,r):=M_1(u),~~~\inf_{0<r<\infty}\left|Q(u,r)\right|:=Q_1(u).
	\end{align}
	We now assume
	\begin{align}\label{final mass charge}
	\lim_{u\rightarrow\infty} M_1(u) = M_2 <\infty,~~~\lim_{u\rightarrow\infty}Q_1(u)=Q_2<\infty,
	\end{align}
	where $M_2$ and $Q_2$ denote the final local mass and final local charge, respectively, and both are positive numbers. Therefore, for all $u\geq u_0$ there exists $u_0>0$ such that $M_1(u)+Q_1(u)<\infty$. Then we define the real-valued function $R(u)=R(M_1(u),Q_1(u))$ on $[u_0,\infty)$ by
	\begin{align}
	R(u):=
	\begin{cases}
	M_1 + \sqrt{M_1^2-Q_1^2},&~~~\mathrm{if}~~~M_1>Q_1,\\
	M_1,&~~~\mathrm{if}~~~M_1=Q_1,\\
	0, &~~~\mathrm{if}~~~M_1<Q_1.\\
	\end{cases}
	\end{align}
	Setting
	\begin{align}
	R_0:=\lim_{u\rightarrow\infty} R(u)=R(M_2,Q_2).
	\end{align}
	Finally, we shall prove the following Theorem.
	\begin{theorem}
		Let $D\geq 4$. Suppose that the assumption (\ref{final mass charge}) holds. For a given positive constant $(D-2)$-spatial Ricci scalar of compact manifold $\hat{R}(\hat{\sigma})$, the time-like line $r=r_0$ is complete toward future if $r_0>R_0$.
	\end{theorem}
	\begin{proof}
		From (\ref{g tilde}), we have $\log \tilde{g}(u,r)\rightarrow 0$ as $r\rightarrow\infty$ at each $u$. Thus, we obtain
		\begin{align}
		-\log\tilde{g}(u,r_0)=&\int_{r_0}^{\infty}\frac{1}{\tilde{g}}\frac{\partial \tilde{g}}{\partial r}\mathrm{d}r\nonumber\\
		=&\int_{r_0}^{\infty}\left[\frac{2(D-3)M}{r^{D-2}}-\frac{(D-4)}{(D-2)}\frac{\hat{R}(\hat{\sigma})}{r}+\frac{16\pi r V(|\tilde{h}|)}{(D-2)}-\frac{4}{(D-2)Vol^2(\Sigma^{D-2})}\frac{Q^2}{r^{2D-5}}\right] \nonumber\\
		&\times\left[\frac{2\hat{R}(\hat{\sigma})}{(D-2)^2}-\frac{2M}{r^{D-3}}+\frac{2}{(D-2)Vol^2(\Sigma^{D-2})}\frac{Q^2}{r^{2(D-3)}}\right]^{-1}\nonumber\\
		\leq&\int_{r_0}^{\infty}\frac{2(D-3)M_1}{r^{D-2}}\left[\frac{2\hat{R}(\hat{\sigma})}{(D-2)^2}-\frac{2M_1}{r^{D-3}}+\frac{2}{(D-2)Vol^2(\Sigma^{D-2})}\frac{Q_1^2}{r^{2(D-3)}}\right]^{-1}\mathrm{d}r\nonumber\\
		=&-\log\left\{f(r_0,M_1,Q_1)\right\}.
		\end{align}
		
		We introduce $u_1>u_0$ such that  $r_0>R(u)$ and $r_0-R_0>|R_0-R(u)|$ for $u>u_1$. From (\ref{g2}) and (\ref{g tilde}) we have $e^{2F}=g\tilde{g}$. Since $\tilde{g}\leq g$, thus
		\begin{align}
		e^{F(u,r_0)}\geq \tilde{g}(u,r_0)\geq f(r_0,M_1,Q_1).
		\end{align}
		Let us introduce $e^{F(u,r_0)}\mathrm{d}u$ as the proper time element along the line $r=r_0$. Since $Q_1(u)\rightarrow Q_2$ and $M_1(u)\rightarrow M_2$ for $u\rightarrow \infty$, there exists $u_1,u_2$ where $u_2>u_1$ such that
		\begin{align}
		\int_{u_1}^{u_2}e^{F(u,r_0)}\mathrm{d}u\geq&\int_{u_1}^{u_2}f(r_0,M_1,Q_1)\mathrm{d}u\nonumber\\
		\geq&\left(r_0-R_0\right)(u_2-u_1)\rightarrow \infty
		\end{align}
		as $u_2\rightarrow\infty$. The above equation ensures that the points of the hypersurfaces are at future timelike infinity for $r_0> R_0$. Finally, we have shown that spacetime is complete along time-like lines outside the region determined by the final local mass and the final local charge. This is the end of the proof.
	\end{proof}
	
	\section{Appendix}\label{Sec6}
	\subsection{The Einstein equations}\label{Appendix0}
	We write down the Ricci tensor related to metric (\ref{metric}) as follows
	\begin{align}
	R_{00}=&-e^{F-G}\left[\partial_1\partial_0 (F+G) +\frac{(D-2)}{r}\partial_0 G\right]\nonumber\\
	&+ e^{2(F-G)}\left[(\partial_1F)^2+\partial_1(\partial_1 F)-\partial_1F\partial_1 G +\frac{(D-2)}{r}\partial_1 F \right],\\
	R_{01}=&R_{10}=-\partial_1\partial_0(F+G)\nonumber\\
	&+e^{F-G}\partial_1(F-G)\partial_1 F + e^{F-G}\partial_1(\partial_1 F)+\frac{(D-2)}{r}e^{F-G}\partial_1 F,\\
	R_{11}=&\frac{(D-2)}{r}\partial_1(F+G),\\
	R_{ij}=& \hat{R}_{ij}(\hat{\sigma}) - e^{-2G}\left[r\partial_1(F-G)+D-3\right]\hat{\sigma}_{ij},
	\end{align}
	where we define the notation $\partial_0:=\frac{\partial}{\partial u}$, $\partial_1:=\frac{\partial}{\partial r}$, $\hat{R}_{ij}(\hat{\sigma})$ is the $(D-2)$-spatial Ricci tensor of compact manifold, and $i,j=1,2,...,D-1$. It is of interest to write down the scalar curvature
	\begin{align}
	R=&\frac{1}{r^2}\hat{R}(\hat{\sigma})+2e^{-(F+G)}\partial_1\partial_0 (F+G)+e^{-2G}\left[-2(\partial_1 F)^2-2\partial_1(\partial_1 F)\right.\nonumber\\
	&\left.+2(\partial_1G)(\partial_1 F)-2\frac{(D-2)}{r}\partial_1(F-G)-\frac{(D-2)(D-3)}{r^2}\right],
	\end{align}
	with $\hat{R}(\hat{\sigma}):=\hat{\sigma}^{ij}\hat{R}_{ij}(\hat{\sigma})$ denotes the $(D-2)$-spatial Ricci scalar of compact manifold.
	
	Now, we write the Einstein equation in higher dimensions
	\begin{align}\label{Einstein}
	R_{\mu\nu}=8\pi \left(T_{\mu\nu}-\frac{g_{\mu\nu}}{(D-2)} T\right),
	\end{align}
	with the energy-momentum tensor
	\begin{align}
	T_{\mu\nu}=&-\tilde{F}_{\mu\alpha}\tilde{F}_{\nu\beta}g^{\alpha\beta}+\frac{1}{4}g_{\mu\nu}\tilde{F}_{\rho\alpha}\tilde{F}_{\sigma\beta}g^{\rho\sigma}g^{\alpha\beta}+D_\mu\phi(D_\nu\phi)^*\nonumber\\
	&-\frac{1}{2}g_{\mu\nu}g^{\alpha\beta}D_\alpha\phi (D_\beta\phi)^* + g_{\mu\nu} V(|\phi|),
	\end{align}
	where we have denoted $T$ as the trace of $T_{\mu\nu}$. Thus, the $\{rr\}$ component of (\ref{Einstein}) satisfies
	\begin{align} \label{rr}
	\frac{\partial}{\partial r}(F+G)=\frac{8\pi r}{D-2}\left|\frac{\partial \phi}{\partial r}\right|^2.
	\end{align}
	The solution of (\ref{rr}) satisfies the asymptotic condition $F+G\rightarrow 0$ as $r \rightarrow\infty$ such
	that
	\begin{align}\label{F+G}
	F+G = -\frac{8\pi}{D-2} \int_{r}^{\infty} s\left|\frac{\partial\phi}{\partial s}\right|^2\mathrm{d}s.
	\end{align}
	Also, the $\{ij\}$ component of (\ref{Einstein}) satisfies
	\begin{align}\label{ij}
	\partial_1(F-G)\hat{\sigma}_{ij}=&-\frac{1}{r}\left[-e^{2G}\hat{R}_{ij}(\hat{\sigma})+(D-3)\hat{\sigma}_{ij}\right]\nonumber\\
	&-\frac{2}{(D-2)}r(\partial_1A_0)^2e^{-2F}\hat{\sigma}_{ij}+\frac{16\pi r e^{2G}}{(D-2)}\hat{\sigma}_{ij}V(|\phi|).
	\end{align}
	On the other hand, the $\{00\}$ and $\{01\}$ components of (\ref{Einstein}) yields 
	\begin{align} \label{00-01}
	\frac{(D-2)}{r}\frac{\partial G}{\partial u}=&8\pi \left[\frac{\partial\phi}{\partial u}\frac{\partial\phi^*}{\partial r}-e^{-(F-G)}\left(\frac{\partial\phi}{\partial u}\frac{\partial\phi^*}{\partial u}-iA_0\left(\phi^*\frac{\partial\phi}{\partial u}-\phi\frac{\partial\phi^*}{\partial u}\right)\right)\right.\nonumber\\
	&\left.+iA_0\phi\frac{\partial\phi^*}{\partial r}-(A_0)^2\phi\phi^*\right].
	\end{align}
	As can be seen from (\ref{rr}), (\ref{ij}), and (\ref{00-01}) there are no terms in $F$ and $G$ containing second derivatives of $u$ and $r$. So the metric functions $F$ and $G$ are no longer dynamical fields, but they are just constraints. Throughout this paper we will use (\ref{rr}) and (\ref{ij}) instead of other components for analysis. This is because it is easier and already common in standard literature \cite{Chris1,Chae1,Prisma,Prisma2}.
	
	\subsection{Estimate for (\ref{f})}\label{Appendix1}
	We write the estimate of (\ref{f}) as follows:
	\begin{align}
	|f|\leq&\frac{(D-2)}{4r}\left|\frac{\hat{R}(\hat{\sigma})}{(D-2)}g-\frac{(D-2)}{2}\tilde{g}\right|\left|\tilde{h}\right|+\frac{8\pi g r}{(D-2)}\left|V(|\tilde{h}|)\right|\left|\tilde{h}\right|+\frac{gr}{2}\left|\frac{\partial V(|\tilde{h}|)}{\partial \tilde{h}^*}\right|\nonumber\\
	&+\left|\frac{gQ^2\tilde{h}}{(D-2)r^3}\right|+\left|\frac{igQ\tilde{h}}{2r}\right|\nonumber\\
	\leq& B_1 +B_2 + B_3 + B_4 + B_5
	\end{align}
	\begin{enumerate}
		\item Estimate for $B_1$\\
		Using the definition of norm (\ref{hx}), we calculate
		\begin{align}
		|\tilde{h}|\leq& \frac{1}{r^{\frac{(D-2)}{2}}}\int_{0}^{r}\frac{\|h\|_X}{(1+s+u)^{k-1}}\frac{1}{s^{\frac{(4-D)}{2}}}\mathrm{d}s\nonumber\\
		\leq&\frac{x}{r^{\frac{(D-2)}{2}}}\int_{0}^{r}\frac{1}{(1+s+u)^{\frac{(2k-D+2)}{2}}}\mathrm{d}s\nonumber\\
		\leq&\frac{2x}{(2k-D)}\frac{1}{r^{\frac{(D-2)}{2}}}\left[\frac{1}{(1+u)^{\frac{(2k-D)}{2}}}-\frac{1}{(1+r+u)^{\frac{(2k-D)}{2}}}\right]\nonumber\\
		\leq&\frac{2x}{(2k-D)}\frac{1}{(1+u)^{\frac{(2k-D)}{2}}(1+r+u)^{\frac{(D-2)}{2}}}.
		\end{align}
		Then, we get
		\begin{align}
		\left|h-\frac{(D-2)}{2}\tilde{h}\right|\leq& |h| + \left|\frac{(D-2)}{2}\tilde{h}\right|\nonumber\\
		\leq&\frac{x}{(1+r+u)^{k-1}}+\frac{2x}{(2k-D)}\frac{1}{(1+u)^{\frac{(2k-D)}{2}}(1+r+u)^{\frac{(D-2)}{2}}}\nonumber\\
		\leq& \frac{Cx}{(1+u)^{\frac{(2k-D)}{2}}(1+r+u)^{\frac{(D-2)}{2}}}.
		\end{align}
		From the above estimate, we obtain
		\begin{align}
		|g(u,r)-g(u,r')|\leq&\int_{r'}^{r}\left|\frac{\partial g}{\partial s}(u,s) \right|\mathrm{d}s\leq \frac{8\pi}{(D-2)} \int_{r'}^{r}\frac{1}{s}\left|h-\frac{(D-2)}{2}\tilde{h}\right|^2\mathrm{d}s\nonumber\\
		\leq& \frac{8\pi}{(D-2)^2}\frac{x^2}{(1+u)^{2k-D}}\left[\frac{1}{(1+r'+u)^{D-2}}-\frac{1}{(1+r+u)^{D-2}}\right],
		\end{align}
		and
		\begin{align}
		|(g-\bar{g})(u,r)|\leq& \frac{1}{r}\int_{0}^{r}|g(u,r)-g(u,r')|\mathrm{d}r'
		\nonumber\\
		\leq& \frac{8\pi x^2}{(D-3)(D-2)^2}\frac{1}{(1+u)^{2k-3}(1+r+u)}.
		\end{align}
		Thus we have
		\begin{align}
		\int_{r}^{\infty}\frac{1}{s}\left|h-\frac{(D-2)}{2}\tilde{h}\right|^2\mathrm{d}s\leq\frac{x^2}{(D-2)(1+u)^{2k-D}(1+r+u)^{D-2}},
		\end{align}
		such that
		\begin{align}
		|g (u,r)|=& \exp \left[-\frac{8\pi}{(D-2)}\int_{r}^{\infty}\frac{1}{s}\left|h-\frac{(D-2)}{2}\tilde{h}\right|^2\mathrm{d}s\right]\nonumber\\
		\geq&\exp \left[-\frac{8\pi x^2}{(D-2)^2(1+u)^{2k-D}(1+r+u)^{D-2}}\right].
		\end{align}
		Hence, we can calculate
		\begin{align}
		|\bar{g}|\geq |g| + |g-\bar{g}|\geq\frac{8\pi x^2}{(D-3)(D-2)^2}\frac{1}{(1+u)^{2k-3}(1+r+u)},
		\end{align}
		and
		\begin{align}
		\frac{(D-4)}{2}\frac{\hat{R}(\hat{\sigma})}{r^{D-3}}\int_{0}^{r}|\bar{g}|s^{D-4}\mathrm{d}s\leq \frac{Cx^2}{(1+u)^{2k-3}(1+r+u)}.
		\end{align}
		Setting $p\in[k,\infty)$. To proceed, we calculate
		\begin{align}
		\frac{8\pi}{r^{D-3}}\int_{0}^{r} gs^{D-2}|V(|\tilde{h}|)|\mathrm{d}s\leq \frac{Cx^{p+1}}{(1+u)^{k^2-D}(1+r+u)^{D-3}}.
		\end{align}
		Now, we calculate the estimate for the local charge. Equation (\ref{Qh}) yields
		\begin{align}
		|Q(u,r)|=&\left|Vol(\Sigma^{D-2})i \int_{0}^{r}\left(\tilde{h}^*h-\tilde{h}h^*\right)s^{D-3}~\mathrm{d}s\right|\nonumber\\
		\leq&2Vol(\Sigma^{D-2})\int_{0}^{r}|\tilde{h}||h|s^{D-3}~\mathrm{d}s\nonumber\\
		\leq&\frac{4Vol(\Sigma^{D-2})}{(2k-D)}\frac{x^2}{(1+u)^{\frac{(2k-D)}{2}}} \int_{0}^{r}\frac{s^{D-3}}{(1+s+u)^{\frac{(D+2k-4)}{2}}}\mathrm{d}s\nonumber\\
		\leq& \frac{8Vol (\Sigma^{D-2})}{(2k-D)^2}\frac{x^2 r^{\frac{(2k-D)}{2}}}{(1+u)^{2k-D}(1+r+u)^{\frac{(2k-D)}{2}}}.
		\end{align}
		From the above estimate, we obtain
		\begin{align}
		\frac{1}{Vol^2(\Sigma^{D-2})}\frac{1}{r^{D-3}}\int_{0}^{r}\left|\frac{gQ^2}{s^{D-2}}\right|\mathrm{d}s\leq&\frac{64}{(D-3)(2k-D)^4}\frac{x^4}{(1+u)^{4k-D-3}(1+r+u)^{D-3}}.
		\end{align}
		Finally, from (\ref{g tilde}) we get
		\begin{align}
		\left|\frac{\hat{R}(\hat{\sigma})}{(D-2)}g-\frac{(D-2)}{2}\tilde{g}\right|\leq&|g-\bar{g}|	+\frac{(D-4)}{2}\frac{\hat{R}(\hat{\sigma})}{r^{D-3}}\int_{0}^{r}|\bar{g}|s^{D-4}\mathrm{d}s\nonumber\\
		&+\frac{8\pi}{r^{D-3}}\int_{0}^{r} gs^{D-2}|V(|\tilde{h}|)|\mathrm{d}s\nonumber\\
		&+	\frac{1}{Vol^2(\Sigma^{D-2})}\frac{1}{r^{D-3}}\int_{0}^{r}\left|\frac{gQ^2}{s^{D-2}}\right|\mathrm{d}s\nonumber\\
		\leq& \frac{C(x^2+x^4+x^{p+1})}{(1+u)^{2k-3}(1+r+u)}.
		\end{align}
		Hence,
		\begin{align}
		B_1=\frac{(D-2)}{4r}\left|\frac{\hat{R}(\hat{\sigma})}{(D-2)}g-\frac{(D-2)}{2}\tilde{g}\right|\left|\tilde{h}\right|\leq\frac{C(x^3+x^5+x^{p+2})}{(1+u)^{\frac{6k-D-6}{2}}(1+r+u)^{\frac{D+2}{2}}}.
		\end{align}
		\item Estimate for $B_2$\\
		Using (\ref{potential}) and (\ref{estimasi htilde}), we obtain
		\begin{align}
		B_2=\frac{8\pi g r}{(D-2)}\left|V(|\tilde{h}|)\right||\tilde{h}|\leq\frac{Cx^{p+2}}{(1+u)^{\frac{(2k-D)(k+2)}{2}}(1+r+u)^{\frac{(D-2)(k+2)-2}{2}}}.
		\end{align}
		\item Estimate for $B_3$\\
		Again, the estimate (\ref{potential}) yields
		\begin{align}
		B_3=\frac{gr}{2}\left|\frac{\partial V(|\tilde{h}|)}{\partial \tilde{h}^*}\right|\leq\frac{Cx^p}{(1+u)^{\frac{(2k-D)k}{2}}(1+r+u)^{\frac{(D-2)k-2}{2}}}.
		\end{align}
		\item Estimate for $B_4$\\
		Using (\ref{estimasi htilde}) and (\ref{estimate Q}) we obtain
		\begin{align}
		B_4=\left|\frac{gQ^2\tilde{h}}{(D-2)r^3}\right|\leq \frac{128 Vol^2(\Sigma^{D-2})}{(2k-D)^5}\frac{x^5}{(1+u)^{\frac{5(2k-D)}{2}}(1+r+u)^{\frac{D+4}{2}}}.
		\end{align}
		\item Estimate for $B_5$\\
		Again, using (\ref{estimasi htilde}) and (\ref{estimate Q}) we obtain
		\begin{align}
		B_5=\left|\frac{igQ\tilde{h}}{2r}\right|\leq \frac{16Vol(\Sigma^{D-2})}{(2k-D)^3}\frac{x^3}{(1+u)^{\frac{3(2k-D)}{2}}(1+r+u)^{\frac{D}{2}}}.
		\end{align}
	\end{enumerate}
	Finally, we obtain
	\begin{align}
	|f|\leq \frac{C(x^3+x^5+x^p+x^{p+2})}{(1+u)^{\frac{(2k-D)k}{2}}(1+r+u)^{\frac{D}{2}}}.
	\end{align}	
	
	\subsection{Estimate for (\ref{f1})}\label{Appendix2}
	We write the estimate of (\ref{f1}) as follows:
	\begin{align}
	|f_1|\leq& \left\{\left|\frac{1}{2r}\left(\frac{\hat{R}(\hat{\sigma})}{(D-2)}\frac{\partial g}{\partial r}-\frac{(D-2)}{2}\frac{\partial\tilde{g}}{\partial r}\right)\right|+\left|\frac{1}{2r^2}\left(\frac{\hat{R}(\hat{\sigma})}{(D-2)}g-\frac{(D-2)}{2}\tilde{g}\right)\right|\right.\nonumber\\
	&+\left|\frac{8\pi g r}{(D-2)} \frac{\partial V(|\tilde{h}|)}{\partial \tilde{h}^*} \frac{\partial \tilde{h}^*}{\partial r}\right|+\left|\frac{8\pi gV(|\tilde{h}|)}{(D-2)}\right|+\left|\frac{8\pi rV(|\tilde{h}|)}{(D-2)}\frac{\partial g}{\partial r}\right|+\left|\frac{Q^2}{(D-2)r^3}\frac{\partial g}{\partial r}\right|\nonumber\\
	&\left.+\left|\frac{2gQ}{(D-2)r^3}\frac{\partial Q}{\partial r}\right|+\left|\frac{3gQ^2}{(D-2)r^4}\right|\right\}\left(\left|\mathcal{F}\right|+\frac{(D-2)}{2}\left|\tilde{h}\right|\right)\nonumber\\
	& +\left\{\left|\frac{1}{2r} \left(\frac{\hat{R}(\hat{\sigma})}{(D-2)}g-\frac{(D-2)}{2}\tilde{g}\right)\right|+\left|\frac{8\pi g rV(|\tilde{h}|)}{(D-2)}\right|+\left|\frac{g r }{2}\frac{\partial^2V(|\tilde{h}|)}{(\partial\tilde{h}^*)^2}\right|\right.\nonumber\\
	&\left.+\left|\frac{gQ^2}{2r^3}\right|+\left|\frac{igQ}{2r}\right|\right\}\left|\frac{\partial \tilde{h}}{\partial r}\right|+\left[\left|\frac{igQ}{2r^2}\right|+\left|\frac{ig}{2r}\frac{\partial Q}{\partial r}\right|+\left|\frac{iQ}{2r}\frac{\partial g}{\partial r}\right|\right]\left|\tilde{h}\right|\nonumber\\
	&+\left|\frac{r}{2}\frac{\partial g}{\partial r}+\frac{g}{2}\right|\left|\frac{\partial V(|\tilde{h}|)}{\partial\tilde{h}^*}\right|+\left|i\frac{\partial A_0}{\partial r}\right|\left|\mathcal{F}\right|\nonumber\\
	=& (E_1+E_2+E_3+E_4+E_5+E_6+E_7+E_8)\left(\left|\mathcal{F}\right|+\frac{(D-2)}{2}\left|\tilde{h}\right|\right)\nonumber\\
	&+ (E_9+E_{10}+E_{11}+E_{12}+E_{13})\left|\frac{\partial \tilde{h}}{\partial r}\right| + \left(E_{14}+E_{15}+E_{16}\right)\left|\tilde{h}\right|+E_{17}+E_{18}
	\end{align}
	\begin{enumerate}
		\item Estimate for $E_1$\\
		The estimate (\ref{g-gbar k}) yields
		\begin{align}
		\frac{|g-\bar{g}|}{r}\leq\frac{Cx^2}{(1+u)^{2k-3}(1+r+u)^2}.
		\end{align}
		Using (\ref{estimate gbar}), we get
		\begin{align}
		\frac{\hat{R}(\hat{\sigma})(D-4)(D-3)}{(D-2)r^{D-2}}\int_{0}^{r}|\bar{g}|s^{D-4}\mathrm{d}s\leq\frac{Cx^2}{(1+u)^{2k-3}(1+r+u)^2},
		\end{align}
		and
		\begin{align}
		\frac{(D-4)}{(D-2)}\hat{R}(\hat{\sigma})\frac{1}{r}|\bar{g}|\leq\frac{8\pi x^2}{(D-3)(D-2)^2(1+u)^{2k-3}(1+r+u)^2}.
		\end{align}
		Then, using (\ref{potential}) and (\ref{htilde}) we obtain
		\begin{align}
		\frac{(D-3)16\pi}{(D-2)r^{D-2}}\int_{0}^{r} \left|g\right|s^{D-2}\left|V(|\tilde{h}|)\right|\mathrm{d}s\leq \frac{Cx^{p+1}}{(1+u)^{k^2-D}(1+r+u)^{D-2}},
		\end{align}
		and
		\begin{align}
		\frac{16\pi gr}{(D-2)}|V(|\tilde{h}|)|\leq\frac{Cx^{p+1}r}{(1+u)^{\frac{(2k-D)(k+1)}{2}}(1+r+u)^{\frac{(D-2)(k+1)}{2}}}.
		\end{align}
		Next, using (\ref{estimate integral Q}) we get
		\begin{align}
		\frac{2(D-3)}{(D-2)Vol^2(\Sigma^{D-2})}\frac{1}{r^{D-2}}\int_{0}^{r}\left|\frac{gQ^2}{s^{D-2}}\right|\mathrm{d}s\leq \frac{Cx^4r}{(1+u)^{4k-D-3}(1+r+u)^{D-1}}.
		\end{align}
		And using (\ref{estimate Q}) we obtain
		\begin{align}
		\frac{2}{(D-2)Vol^2(\Sigma^{D-2})}\left|\frac{gQ^2}{r^{2D-5}}\right|\leq \frac{Cx^4r} {(1+u)^{2(2k-D)}(1+r+u)^{2D-4}}.
		\end{align}
		Combining all above estimates yields
		\begin{align}
		\left|\frac{\partial\tilde{g}}{\partial r}\right|\leq \frac{C(x^2+x^{4}+x^{p+1})r}{(1+u)^{2k-3}(1+r+u)^3}.
		\end{align}
		On the other hand, from (\ref{g2}) we have
		\begin{align}\label{dgdr}
		\left|\frac{\partial g}{\partial r}\right|\leq\frac{Cx^2r}{(1+u)^{2k-D}(1+r+u)^{D}}.
		\end{align}
		Finally, we obtain
		\begin{align}
		E_1=\left|\frac{1}{2r}\left(\frac{\hat{R}(\hat{\sigma})}{(D-2)}\frac{\partial g}{\partial r}-\frac{(D-2)}{2}\frac{\partial\tilde{g}}{\partial r}\right)\right|\leq\frac{C(x^2+x^4+x^{p+1})}{(1+u)^{2k-D}(1+r+u)^3}.
		\end{align}
		
		\item Estimate for $E_2$\\
		Using (\ref{g - gtilde}) we obtain
		\begin{align}
		E_2=\left|\frac{1}{2r^2}\left(\frac{\hat{R}(\hat{\sigma})}{(D-2)}g-\frac{(D-2)}{2}\tilde{g}\right)\right|\leq\frac{C(x^2+x^4+x^{p+1})}{(1+u)^{2k-3}(1+r+u)^{3}}.
		\end{align}
		
		\item Estimate for $E_3$\\
		Using the relation $\left|\frac{\partial \tilde{h}}{\partial r}\right|=\frac{\left|h-\frac{(D-2)}{2}\tilde{h}\right|}{r}$, we get
		\begin{align}
		E_3 = \left|\frac{8\pi g r}{(D-2)} \frac{\partial V(|\tilde{h}|)}{\partial \tilde{h}^*} \frac{\partial \tilde{h}^*}{\partial r}\right|\leq\frac{Cx^{p+1}}{(1+u)^{\frac{(2k-D)(k+1)}{2}}(1+r+u)^{\frac{(D-2)(k+1)}{2}}}.
		\end{align}
		
		\item Estimate for $E_4$\\
		From (\ref{potential}) we obtain
		\begin{align}
		E_4=\left|\frac{8\pi g V(|\tilde{h}|)}{(D-2)}\right|
		\leq \frac{Cx^{p+1}}{(1+u)^{\frac{(2k-D)(k+1)}{2}}(1+r+u)^{\frac{(D-2)(k+1)}{2}}}.
		\end{align}
		
		\item Estimate for $E_5$\\
		Combination of (\ref{potential}) and (\ref{dgdr}) yields
		\begin{align}
		E_5 = \left|\frac{8\pi rV(|\tilde{h}|)}{(D-2)}\frac{\partial g}{\partial r}\right|\leq\frac{Cx^{p+3}}{(1+u)^{3(2k-D)}(1+r+u)^{\frac{(D-2)(k+3)}{2}}}.
		\end{align}
		
		\item Estimate for $E_6$\\
		Combination of (\ref{estimate Q}) and (\ref{dgdr}) produce
		\begin{align}
		E_6=\left|\frac{Q^2}{(D-2)r^3}\frac{\partial g}{\partial r}\right|\leq \frac{Cx^6}{(1+u)^{3(2k-D)}(1+r+u)^{D+2}}.
		\end{align}
		
		\item Estimate for $E_7$\\
		From (\ref{estimate Q}) we obtain
		\begin{align}
		\left|\frac{\partial Q}{\partial r}\right|\leq&\left|Vol(\Sigma^{D-2})i \left(\tilde{h}^*h-\tilde{h}h^*\right)r^{D-3}\right|\nonumber\\
		\leq& 2Vol(\Sigma^{D-2}) \left|\tilde{h}\right|\left|h\right|r^{D-3}\nonumber\\
		\leq&\frac{Cx^2}{(1+u)^{\frac{(2k-D)}{2}}(1+r+u)^{\frac{2k-D+2}{2}}}.
		\end{align}
		Thus, we have		
		\begin{align}
		E_7=\left|\frac{2gQ}{(D-2)r^3}\frac{\partial Q}{\partial r}\right|\leq \frac{Cx^4}{(1+u)^{\frac{3(2k-D)}{2}}(1+r+u)^{\frac{(2k-D+8)}{2}}}.
		\end{align}
		
		\item Estimate for $E_8$\\
		From (\ref{estimate Q}) we obtain
		\begin{align}
		E_8=\left|\frac{3gQ^2}{(D-2)r^4}\right|\leq\frac{Cx^4} {(1+u)^{2(2k-D)}(1+r+u)^4}.
		\end{align}
		
		\item Estimate for $E_9$\\
		Using (\ref{g - gtilde}) we obtain
		\begin{align}
		E_9=\left|\frac{1}{2r}\left(\frac{\hat{R}(\hat{\sigma})}{(D-2)}g-\frac{(D-2)}{2}\tilde{g}\right)\right|\leq\frac{C(x^2+x^4+x^{p+1})}{(1+u)^{2k-3}(1+r+u)^2}.
		\end{align}
		
		\item Estimate for $E_{10}$\\
		The estimate (\ref{potential}) yields
		\begin{align}
		E_{10}=\left|\frac{8\pi g r V(|\tilde{h}|)}{(D-2)}\right|\leq\frac{C x^{p+1}}{(1+u)^{\frac{(2k-D)(k+1)}{2}}(1+r+u)^{\frac{(D-2)(k+1)}{2}-1}}.
		\end{align}
		
		\item Estimate for $E_{11}$\\
		Again, we use the relation (\ref{potential}) to produce
		\begin{align}
		E_{11} =\left|\frac{g r }{2}\frac{\partial^2V(|\tilde{h}|)}{(\partial\tilde{h}^*)^2}\right|\leq \frac{C x^{p-1}}{(1+u)^{\frac{(2k-D)(k-1)}{2}}(1+r+u)^{\frac{(D-2)(k-1)}{2}-1}}.
		\end{align}
		
		\item Estimate for $E_{12}$\\
		As previously we obtain
		\begin{align}
		E_{12}=\left|\frac{gQ^2}{2r^3}\right|\leq \frac{Cx^4} {(1+u)^{2(2k-D)}(1+r+u)^3}.
		\end{align}
		
		\item Estimate for $E_{13}$\\
		We have
		\begin{align}
		E_{13}=\left|\frac{igQ}{2r}\right|\leq\frac{Cx^2}{(1+u)^{2k-D}(1+r+u)}.
		\end{align}
		
		\item Estimate for $E_{14}$\\
		We easily get
		\begin{align}
		E_{14}=\left|\frac{igQ}{2r^2}\right|\leq \frac{Cx^2}{(1+u)^{2k-D}(1+r+u)^2}.
		\end{align}
		
		\item Estimate for $E_{15}$\\
		We can calculate
		\begin{align}
		E_{15}=\left|\frac{ig}{2r}\frac{\partial Q}{\partial r}\right|\leq \frac{Cx^2}{(1+u)^{\frac{(2k-D)}{2}}(1+r+u)^{\frac{(2k-D+4)}{2}}}.
		\end{align}
		\item Estimate for $E_{16}$\\
		It is straighforward to show
		\begin{align}
		E_{16}=\left|\frac{iQ}{2r}\frac{\partial g}{\partial r}\right|\leq \frac{Cx^4}{(1+u)^{2(2k-D)}(1+r+u)^D}.
		\end{align}
		\item Estimate for $E_{17}$\\
		Using (\ref{potential}) we get
		\begin{align}
		E_{17} =\left|\frac{r}{2}\frac{\partial g}{\partial r}+\frac{g}{2}\right|\left|\frac{\partial V(|\tilde{h}|)}{\partial\tilde{h}^*}\right|\leq\frac{C (x^p+x^{p+2})}{(1+u)^{\frac{(2k-D)k}{2}}(1+r+u)^{\frac{(D-2)k}{2}}}.
		\end{align}
		\item Estimate for $E_{18}$\\
		From (\ref{A0}) we obtain
		\begin{align}
		E_{18}=\left|i\frac{\partial A_0}{\partial r}\right|\left|\mathcal{F}\right|=\frac{1}{Vol(\Sigma^{D-2})}\frac{g|Q|}{r^{D-2}}\left|\mathcal{F}\right|\leq \frac{Cx^2}{(1+u)^{2k-D}(1+r+u)^{D-2}}\left|\mathcal{F}\right|.
		\end{align}
	\end{enumerate}
	Finally, we obtain
	\begin{align}
	|f_1|\leq\frac{C(x^2+x^4+x^6+x^{p+1}+x^{p+3})}{(1+u)^{2k-D}(1+r+u)^{D-2}}|\mathcal{F}|+\frac{C(x^3+x^5+x^7+x^p+x^{p+2}+x^{p+4})}{(1+u)^{2k-D}(1+r+u)^\frac{(D+2)}{2}}.
	\end{align}
	
	\subsection{Estimate for (\ref{phi tilde})}\label{Appendix3}
	We write the estimate of (\ref{phi tilde}) as follows:
	\begin{align}
	|\tilde{\varphi}|\leq&\frac{1}{2}\left|\tilde{g}_1-\tilde{g}_2\right|\left|\mathcal{G}_2\right|+\frac{1}{2r}\left|\frac{\hat{R}(\hat{\sigma})}{(D-2)}g_1-\frac{(D-2)}{2}\tilde{g}_1\right|\frac{(D-2)}{2}\left|\tilde{h}_1 -\tilde{h}_2\right|\nonumber\\
	&+\frac{1}{2r}\left|\frac{\hat{R}(\hat{\sigma})}{(D-2)}(g_1-g_2)-\frac{(D-2)}{2}(\tilde{g}_1-\tilde{g}_2)\right|\left|\left|\mathcal{F}_2\right|+\frac{(D-2)}{2}\left|\tilde{h}_2\right|\right|\nonumber\\
	&+\frac{8\pi r}{(D-2)}|g_1-g_2||\mathcal{F}_1||V(|\tilde{h}_1|)|+4\pi r|g_1-g_2|\left|\tilde{h}_1\right| |V(|\tilde{h}_1|)|\nonumber\\
	&+\frac{8\pi r g_2}{(D-2)} \left|V(|\tilde{h}_1|)-V(|\tilde{h}_2|)\right||\mathcal{F}_2|+4\pi r g_2 \left|\tilde{h}_1 - \tilde{h}_2\right||V(|\tilde{h}_2|)|\nonumber\\
	&+4\pi r g_2\left|\tilde{h}_1\right| \left|V(|\tilde{h}_1|)-V(|\tilde{h}_2|)\right|+\frac{r}{2}|g_1-g_2|\left|\frac{\partial V(|\tilde{h}_1|)}{\partial \tilde{h}_1^*}\right|\nonumber\\
	&+\frac{rg_2}{2}\left|\frac{\partial V(|\tilde{h}_1|)}{\partial \tilde{h}_1^*}-\frac{\partial V(|\tilde{h}_2|)}{\partial \tilde{h}_2^*}\right|+\frac{g_1\left|Q_1^2-Q_2^2\right|}{(D-2)r^3}\left|\left|\mathcal{F}_1\right|+\frac{(D-2)}{2}\left|\tilde{h}_1\right|\right|\nonumber\\
	&+ \frac{g_1\left|Q_2^2\right|}{2r^3}\left|\tilde{h}_1-\tilde{h}_2\right|+\frac{|g_1-g_2|}{(D-2)r^3}\left|Q_2^2\right|\left|\left|\mathcal{F}_2\right|+\frac{(D-2)}{2}\left|\tilde{h}_2\right|\right|\nonumber\\
	&+\frac{g_1\left|Q_1-Q_2\right|}{2r}\left|\tilde{h}_1\right|+\frac{|Q_2||g_1-g_2|}{2r}\left|\tilde{h}_1\right|+\frac{|Q_2|g_2}{2r}\left|\tilde{h}_1-\tilde{h}_2\right|+\left|\mathcal{F}_2\right|\left|{A_0}_1-{A_0}_2\right|\nonumber\\
	=& P_1 + P_2 + P_3 + ... + P_{15} + P_{16} + P_{17}.
	\end{align}
	\begin{enumerate}
		\item Estimate for $P_1$\\
		Let us consider $\|h_1-h_2\|_Y=y$.
		\begin{align}
		|\tilde{h}_1 - \tilde{h}_2|\leq\frac{2y}{(2k-D)}\frac{1}{(1+u)^{\frac{(2k-D)}{2}}(1+r+u)^{\frac{(D-2)}{2}}}.
		\end{align}
		Then, we calculate
		\begin{align}
		|h_1-h_2-(\tilde{h}_1-\tilde{h}_2)|\leq& |h_1-h_2|+|\tilde{h}_1-\tilde{h}_2|\nonumber\\
		\leq&\frac{Cy}{(1+u)^{\frac{(2k-D)}{2}}(1+r+u)^{\frac{(D-2)}{2}}},
		\end{align}
		and
		\begin{align}
		\left||h_1-\tilde{h}_1|^2-|h_2-\tilde{h}_2|^2\right|\leq& \left|(h_1-h_2)-(\tilde{h}_1-\tilde{h}_2) \right|\left(|h_1-\tilde{h}_1|+|h_2-\tilde{h}_2|\right)\nonumber\\
		\leq& \frac{Cxy}{(1+u)^{2k-D}(1+r+u)^{D-2}}.
		\end{align}
		As a consequence, we obtain
		\begin{align}
		|g_1-g_2|\leq&\frac{8\pi}{(D-2)} \int_{r}^{\infty}\frac{1}{s}\left||h_1-\tilde{h}_1|^2-|h_2-\tilde{h}_2|^2\right|\mathrm{d}s\nonumber\\
		\leq& \frac{Cxy}{(1+u)^{2k-D}(1+r+u)^{D-2}}.
		\end{align}
		Then, we calculate
		\begin{align}
		\tilde{h}_1^{p+1} - \tilde{h}_2^{p+1}=(\tilde{h}_1-\tilde{h}_2)\int_{0}^{1}\left(t\tilde{h}_1+(1-t)\tilde{h}_2\right)\left|t\tilde{h}_1+(1-t)\tilde{h}_2\right|^{p-1}\mathrm{d}t,
		\end{align}
		to produce
		\begin{align}
		\left|\tilde{h}_1^{p+1}-\tilde{h}_2^{p+1}\right|\leq& |\tilde{h}_1-\tilde{h}_2|\left(|\tilde{h}_1|+|\tilde{h}_2|\right)^p\nonumber\\
		\leq&\frac{2^{2p+1}x^py}{(2k-D)^{k+1}(1+u)^{\frac{(2k-D)(k+1)}{2}}(1+r+u)^{\frac{(D-2)(k+1)}{2}}}.
		\end{align}
		It is straighforward to show
		\begin{align}
		\frac{16\pi}{(D-2)}\frac{1}{r^{D-3}}\int_{0}^{r} g_1s^{D-2}\left|V(|\tilde{h}_1|)-V(|\tilde{h}_2|)\right|\mathrm{d}s\leq \frac{Cx^py}{(1+u)^{k^2-D}(1+r+u)^{D-3}},
		\end{align}
		and
		\begin{align}
		\frac{16\pi}{(D-2)}\frac{1}{r^{D-3}}\int_{0}^{r}s^{D-2}|V(|\tilde{h}_1|)||g_1-g_2|\mathrm{d}s\leq \frac{Cx^{p+2}y}{(1+u)^{\frac{(2k+1)k-2(D-1)}{2}}(1+r+u)^{D-3}}.
		\end{align}
		Using the mean value formula we obtain
		\begin{align}
		|\bar{g}_1-\bar{g}_2|\leq\frac{1}{r}\int_{0}^{r}|g_1-g_2|\mathrm{d}s\leq\frac{Cxy}{(1+u)^{2k-3}(1+r+u)},
		\end{align}
		such that
		\begin{align}
		\frac{(D-4)}{(D-2)}\frac{\hat{R}(\hat{\sigma})}{r^{D-3}}\int_{0}^{r}\left|\bar{g}_1-\bar{g}_2\right|s^{D-4}\mathrm{d}s\leq \frac{Cxy}{(1+u)^{2k-3}(1+r+u)}.
		\end{align}
		To proceed, we use the estimate (\ref{estimate Q}) to show
		\begin{align}
		|Q_1-Q_2|\leq&Vol(\Sigma^{D-2}) \int_{0}^{r}|\tilde{h}_1^*h_1 - \tilde{h}_1h_1^*-\tilde{h}_2^*h_2+\tilde{h}_2h^*_2|s^{D-3}~\mathrm{d}s\nonumber\\
		\leq&2Vol(\Sigma^{D-2}) \int_{0}^{r}(|\tilde{h}_1-\tilde{h}_2||h_1|+|h_1-h_2||\tilde{h}_2|)s^{D-3}~\mathrm{d}s\nonumber\\
		\leq&\frac{Cxy}{(1+u)^{\frac{(2k-D)}{2}}}\int_{0}^{r}\frac{s^{D-3}}{(1+s+u)^{\frac{(2k+D-4)}{2}}} \mathrm{d}s\leq\frac{Cxyr^{\frac{(2k-D)}{2}}}{(1+u)^{2k-D}(1+r+u)^{\frac{(2k-D)}{2}}}.
		\end{align}
		Thus we obtain
		\begin{align}
		\frac{1}{r^{D-3}}\int_{0}^{r}\frac{|Q_1^2-Q_2^2|}{s^{D-2}}g\mathrm{d}s\leq&\frac{1}{r^{D-3}}\int_{0}^{r}\frac{1}{s^{D-2}}|Q_1-Q_2|(|Q_1|+|Q_2|)\mathrm{d}s\nonumber\\
		\leq&\frac{Cx^3y}{(1+u)^{2(2k-D)+D-3}(1+r+u)^{D-3}},
		\end{align}
		and
		\begin{align}
		\frac{1}{r^{D-3}}\int_{0}^{r}\frac{|g_1-g_2|}{s^{D-2}}|Q_1|^2\mathrm{d}s\leq\frac{Cx^5y}{(1+u)^{6k-D-5}(1+r+u)^{D-3}}.
		\end{align}
		
		From the definition (\ref{g tilde}) we get
		\begin{align}
		\left|\tilde{g}_1-\tilde{g}_2\right|\leq& \frac{\hat{R}(\hat{\sigma})}{(D-2)}\left|\bar{g}_1-\bar{g}_2\right|+\frac{(D-4)}{(D-2)}\frac{\hat{R}(\hat{\sigma})}{r^{D-3}}\int_{0}^{r}\left|\bar{g}_1-\bar{g}_2\right|s^{D-4}\mathrm{d}s\nonumber\\
		&+ \frac{16\pi}{(D-2)}\frac{1}{r^{D-3}}\int_{0}^{r} g_1s^{D-2}\left|V(|\tilde{h}_1|)-V(|\tilde{h}_2|)\right|\mathrm{d}s\nonumber\\
		&+\frac{16\pi}{(D-2)}\frac{1}{r^{D-3}}\int_{0}^{r}s^{D-2}|V(|\tilde{h}_1|)||g_1-g_2|\mathrm{d}s\nonumber\\
		&+\frac{2}{(D-2)Vol^2(\Sigma^{D-2})}\frac{1}{r^{D-3}}\int_{0}^{r}\frac{|Q_1^2-Q_2^2|}{s^{D-2}}g\mathrm{d}s\nonumber\\
		&+\frac{2}{(D-2)Vol^2(\Sigma^{D-2})}\frac{1}{r^{D-3}}\int_{0}^{r}\frac{|g_1-g_2|}{s^{D-2}}|Q_1|^2\mathrm{d}s\nonumber\\
		\leq&\frac{C(x+x^3+x^5+x^p+x^{p+2})y}{(1+u)^{2k-3}(1+r+u)}.
		\end{align}
		Then, using (\ref{estimate Gcurl}) we obtain
		\begin{align}
		P_1=&\frac{1}{2}|\tilde{g}_1-\tilde{g}_2||\mathcal{G}_2|\nonumber\\
		\leq&\frac{C(x+x^3+x^5+x^p+x^{p+2})y}{(1+u)^{2k-3}(1+r+u)}\exp\left[C(x^2+x^4+x^{p+1})\right]\nonumber\\
		&\times \frac{C(d+x^3+x^5+x^7+x^p+x^{p+2}+x^{p+4})(1+x^2+x^4+x^6+x^{p+1}+x^{p+3})}{\kappa^k(1+r_1+u_1)^k}.
		\end{align}
		Since $1+r(u)+u\geq \frac{1}{2}\kappa(1+r_1+u_1)$, we have
		\begin{align}
		P_1\leq \frac{Cy\alpha(x)}{(1+u)^{2k-3}(1+r+u)^{k+1}},
		\end{align}
		where $\alpha(x)=C(x+x^3+x^5+x^p+x^{p+2})(d+x^3+x^5+x^7+x^p+x^{p+2}+x^{p+4})(1+x^2+x^4+x^6+x^{p+1}+x^{p+3})\exp\left[C(x^2+x^4+x^{p+1})\right]$.
		
		\item Estimate for $P_2$\\
		We use (\ref{g - gtilde}) and (\ref{h1-h2}) to calculate the following estimates
		\begin{align}
		P_2=&\frac{1}{2r}\left|\frac{\hat{R}(\hat{\sigma})}{(D-2)}g_1-\frac{(D-2)}{2}\tilde{g}_1\right|\frac{(D-2)}{2}\left|\tilde{h}_1 -\tilde{h}_2\right|\nonumber\\
		\leq&\frac{C(x^2+x^4+x^{p+1})y}{(1+u)^{\frac{6(k-1)-D}{2}}(1+r+u)^{\frac{(D+2)}{2}}}.
		\end{align}
		\item Estimate for $P_3$\\
		We firstly calculate
		\begin{align}
		\left|\frac{\hat{R}(\hat{\sigma})}{(D-2)}(g_1-g_2)-\frac{\hat{R}(\hat{\sigma})}{2}(\bar{g}_1-\bar{g}_2)\right|
		\leq&\left|g_1-g_2-(\bar{g}_1-\bar{g}_2) \right|\nonumber\\
		\leq&\frac{1}{r}\int_{0}^{r}\int_{r'}^{r}\left|\frac{\partial}{\partial r}(g_1-g_2)\right|\mathrm{d}s~\mathrm{d}r'\nonumber\\
		\leq&\frac{C}{r}\int_{0}^{r}\int_{r'}^{r}\frac{1}{s}\left||h_1-\tilde{h}_1|^2-|h_2-\tilde{h}_2|^2\right|\mathrm{d}s\mathrm{d}r'\nonumber\\
		\leq& \frac{Cxy}{r(1+u)^{2k-D}}\int_{0}^{r}\int_{r'}^{r}\frac{\mathrm{d}s}{(1+s+u)^{D-1}}\mathrm{d}r'\nonumber\\
		\leq& \frac{Cxy}{(1+u)^{2k-3}(1+r+u)^{D-3}}.
		\end{align}
		Then, we obtain
		\begin{align}
		&\frac{1}{2r}\left|\frac{\hat{R}(\hat{\sigma})}{(D-2)}(g_1-g_2)-\frac{(D-2)}{2}(\tilde{g}_1-\tilde{g}_2)\right|\leq\frac{1}{2r}	\left|\frac{\hat{R}(\hat{\sigma})}{(D-2)}(g_1-g_2)-\frac{\hat{R}(\hat{\sigma})}{2}(\bar{g}_1-\bar{g}_2)\right|\nonumber\\
		&+\frac{(D-4)}{4}\frac{\hat{R}(\hat{\sigma})}{r^{D-2}}\int_{0}^{r}|\bar{g}_1-\bar{g}_2|s^{D-4}\mathrm{d}s+\frac{4\pi}{r^{D-2}} \int_{0}^{r}g_1s^{D-2}\left|V(|\tilde{h}_1|)-V(|\tilde{h}_2|)\right|\mathrm{d}s\nonumber\\
		&+\frac{4\pi}{r^{D-2}}\int_{0}^{r}s^{D-2}|V(|\tilde{h}_1|)||g_1-g_2|\mathrm{d}s+\frac{1}{2Vol^2(\Sigma^{D-2})}\frac{1}{r^{D-2}}\int_{0}^{r}\frac{g_1}{s^{D-2}}|Q_1^2-Q_2^2|\mathrm{d}s\nonumber\\
		&+\frac{1}{2Vol^2(\Sigma^{D-2})}\frac{1}{r^{D-2}}\int_{0}^{r}\frac{|g_1-g_2|}{s^{D-2}}|Q_1|^2\mathrm{d}s\leq\frac{C(x+x^3+x^5+x^p+x^{p+2})y}{(1+u)^{2k-3}(1+r+u)^{D-2}}.
		\end{align}
		Thus,
		\begin{align}
		P_3=&\frac{1}{2r}\left|\frac{\hat{R}(\hat{\sigma})}{(D-2)}(g_1-g_2)-\frac{(D-2)}{2}(\tilde{g}_1-\tilde{g}_2)\right|\left|\left|\mathcal{F}_2\right|+\frac{(D-2)}{2}\left|\tilde{h}_2\right|\right|\nonumber\\
		\leq& \frac{C(x+x^3+x^5+x^p+x^{p+2})y}{(1+u)^{2k-3}(1+r+u)^{D-2}}\left[\frac{C(d+x^3+x^5+x^p+x^{p+2})}{\kappa^{k-1}(1+r_1+u_1)^{k-1}}\right.\nonumber\\
		&\left.\times \exp\left[C(x^2+x^4+x^{p+1})\right] +\frac{Cx}{(1+u)^{\frac{(2k-D)}{2}}(1+r+u)^{\frac{(D-2)}{2}}}\right].
		\end{align}
		Since $1+r(u)+u\geq \frac{1}{2}\kappa(1+r_1+u_1)$, we obtain
		\begin{align}
		P_3\leq&\frac{C y \beta(x)}{(1+u)^{2k-3}(1+r+u)^{D+k-3}}+\frac{C(x^2+x^4+x^6+x^{p+1}+x^{p+3})y}{(1+u)^{\frac{6k-D-6}{2}}(1+r+u)^{\frac{3(D-2)}{2}}},\nonumber\\
		\leq&\frac{C(\beta(x)+x^2+x^4+x^6+x^{p+1}+x^{p+3})y}{(1+u)^{2k-3}(1+r+u)^{\frac{3(D-2)}{2}}}
		\end{align}
		where $\beta(x)=C(x+x^3+x^5+x^p+x^{p+2})(d+x^3+x^5+x^p+x^{p+2})\exp\left[C(x^2+x^4+x^{p+1})\right].$
		
		\item Estimate for $P_4$\\
		Using (\ref{potential}), (\ref{estimate Fcurl}), and (\ref{g1-g2}) we obtain
		\begin{align}
		P_4=&\frac{8\pi r}{(D-2)}|g_1-g_2||\mathcal{F}_1||V(|\tilde{h}_1|)|\nonumber\\
		\leq&\frac{Cx^{p+2}y}{(1+u)^{\frac{(2k-D)(k+3)}{2}}(1+r+u)^{\frac{(D-2)(k+3)}{2}}}\nonumber\\
		&\times\frac{C(d+x^3+x^5+x^p+x^{p+2})\exp\left[C(x^2+x^4+x^{p+1})\right]}{\kappa^{k-1}(1+r_1+u_1)^{k-1}}.
		\end{align}
		Since $1+r(u)+u\geq \frac{1}{2}\kappa(1+r_1+u_1)$, we obtain
		\begin{align}
		P_4\leq\frac{Cy \gamma(x)}{(1+u)^{\frac{(2k-D)(k+3)}{2}}(1+r+u)^{\frac{D(k+3)-8}{2}}},
		\end{align}
		where $\gamma(x)=C x^{p+2}(d+x+x^3+x^5+x^p+x^{p+2})\exp\left[C(x^2+x^4+x^{p+1})\right].$
		
		\item Estimate for $P_5$\\
		Using (\ref{potential}), (\ref{estimasi htilde}), and (\ref{g1-g2}) we obtain
		\begin{align}
		P_5=4\pi r|g_1-g_2||\tilde{h}_1| |V(|\tilde{h}_1|)|	\leq\frac{Cx^{p+3}y}{(1+u)^{\frac{7}{2}(2k-D)}(1+r+u)^{\frac{6(2k-D)}{2}}}.
		\end{align}
		\item Estimate for $P_6$\\
		Combining (\ref{estimate Fcurl}) and (\ref{V1-V2}) yields
		\begin{align}
		P_6=&\frac{8\pi rg_2}{(D-2)} \left|V(|\tilde{h}_1|)-V(|\tilde{h}_2|)\right||\mathcal{F}_2|\nonumber\\
		\leq&\frac{Cx^pyr}{(1+u)^{\frac{(2k-D)(k+1)}{2}}(1+r+u)^{\frac{(D-2)(k+1)}{2}}}\nonumber\\
		&\times\frac{C(d+x^3+x^5+x^p+x^{p+2})\exp\left[C(x^2+x^4+x^{p+1})\right]}{\kappa^{k-1}(1+r_1+u_1)^{k-1}}.
		\end{align}
		Since $1+r(u)+u\geq \frac{1}{2}\kappa(1+r_1+u_1)$, we obtain
		\begin{align}
		P_6\leq\frac{Cy\sigma(x)}{(1+u)^{2(2k-D)}(1+r+u)^{2(D-2)+1}},
		\end{align}
		where $\sigma(x)=Cx^{p}(d+x^3+x^5+x^p+x^{p+2})\exp\left[C(x^2+x^4+x^{p+1})\right]$.
		
		\item Estimate for $P_7$\\
		By (\ref{potential}) and (\ref{h1-h2}) we get
		\begin{align}
		P_7 = 4\pi r |g_2||\tilde{h}_1 - \tilde{h}_2| |V(|\tilde{h}_2|)|\leq\frac{Cx^{p+1}y}{(1+u)^{\frac{5(2k-D)}{2}}(1+r+u)^{\frac{2(2k-D)}{2}}}.
		\end{align}
		\item Estimate for $P_8$\\
		From (\ref{estimasi htilde}) and (\ref{V1-V2}) we obtain
		\begin{align}
		P_8 = 4\pi r |g_2||\tilde{h}_1| \left|V(|\tilde{h}_1|)-V(|\tilde{h}_2|)\right|
		\leq\frac{Cx^{p+1}y}{(1+u)^{\frac{5(2k-D)}{2}}(1+r+u)^{\frac{4(2k-D)}{2}}}.
		\end{align}
		\item Estimate for $P_9$\\
		Using (\ref{potential}) and (\ref{g1-g2}) we obtain
		\begin{align}
		P_9 = \frac{r}{2}|g_1-g_2|\left|\frac{\partial V(|\tilde{h}_1|)}{\partial \tilde{h}_1^*}\right|\leq \frac{Cx^{p+1}y}{(1+u)^{\frac{(2k-D)(k+2)}{2}}(1+r+u)^{\frac{(D-2)(k+2)}{2}}}.
		\end{align}
		\item Estimate for $P_{10}$\\
		We firstly calculate
		\begin{align}
		\left|\frac{\partial V(|\tilde{h}_1|)}{\partial \tilde{h}_1^*}-\frac{\partial V(|\tilde{h}_2|)}{\partial \tilde{h}_2^*}\right|
		\leq& K_0|\tilde{h}_1-\tilde{h}_2|(|\tilde{h}_1|^{p-1}+|\tilde{h}_2|^{p-1})\nonumber\\
		\leq&\frac{C x^{p-1}y}{(1+u)^{\frac{(2k-D)k}{2}}(1+r+u)^{\frac{(D-2)k}{2}}}.
		\end{align}
		Thus,
		\begin{align}
		P_{10}=\frac{r|g_2|}{2}\left|\frac{\partial V(|\tilde{h}_1|)}{\partial \tilde{h}_1^*}-\frac{\partial V(|\tilde{h}_2|)}{\partial \tilde{h}_2^*}\right|\leq\frac{C x^{p-1}y}{(1+u)^{\frac{k(2k-D)}{2}}(1+r+u)^{\frac{k(D-2)}{2}-1}}.
		\end{align}
		\item Estimate for $P_{11}$\\
		We firstly calculate
		\begin{align}
		\left|Q_1^2-Q_2^2\right|=|Q_1-Q_2|(|Q_1|+|Q_2|)
		\leq\frac{Cx^3yr^{2k-D}}{(1+u)^{2(2k-D)}(1+r+u)^{2k-D}}.
		\end{align}
		Thus we obtain
		\begin{align}
		P_{11} =& \frac{g_1}{(D-2)r^3}\left|Q_1^2-Q_2^2\right| \left||\mathcal{F}_1|+\frac{(D-2)}{2}|\tilde{h}_1|\right|\nonumber\\
		\leq&\frac{ Cx^3y}{(1+u)^{2(2k-D)}(1+r+u)^3}\left[\frac{C(d+x^3+x^5+x^p+x^{p+2})}{\kappa^{k-1}(1+r_1+u_1)^{k-1}}\right.\nonumber\\
		&\left.\times\exp\left[C(x^2+x^4+x^{p+1})\right]+\frac{Cx}{(1+u)^{\frac{(2k-D)}{2}}(1+r+u)^{\frac{(D-2)}{2}}}\right]
		\end{align}
		Since $1+r(u)+u\geq \frac{1}{2}\kappa(1+r_1+u_1)$, we get
		\begin{align}
		P_{11}\leq&\frac{C\eta(x)y}{(1+u)^{2(2k-D)}(1+r+u)^{k+2}}+\frac{Cx^4y}{(1+u)^{\frac{5(2k-D)}{2}}(1+r+u)^{\frac{(D+4)}{2}}}\nonumber\\
		\leq&\frac{C(\eta(x)+x^4)y}{(1+u)^{2(2k-D)}(1+r+u)^{\frac{(D+4)}{2}}}.
		\end{align}
		where $\eta(x)=Cx^3(d+x^3+x^5+x^p+x^{p+2})\exp\left[C(x^2+x^4+x^{p+1})\right]$.
		
		\item Estimate for $P_{12}$\\
		Combinations of (\ref{estimate Q}) and (\ref{h1-h2}) yields
		\begin{align}
		P_{12}=\frac{|g_1||Q_2|^2}{(D-2)r^3}|\tilde{h}_1-\tilde{h}_2|\leq\frac{Cx^4y}{(1+u)^{\frac{5(2k-D)}{2}}(1+r+u)^{\frac{D+4}{2}}}.
		\end{align}
		\item Estimate for $P_{13}$\\
		We have
		\begin{align}
		P_{13}=&\frac{|g_1-g_2|}{(D-2)r^3}\left|Q_2^2\right|\left|\left|\mathcal{F}_2\right|+\frac{(D-2)}{2}\left|\tilde{h}_2\right|\right|\nonumber\\
		\leq&\frac{Cx^5y}{(1+u)^{3(2k-D)}(1+r+u)^{D+1}}\left[\frac{C(d+x^3+x^5+x^p+x^{p+2})}{\kappa^{k-1}(1+r_1+u_1)^{k-1}}\right.\nonumber\\
		&\left.\times\exp\left[C(x^2+x^4+x^{p+1})\right]+\frac{Cx}{(1+u)^{\frac{(2k-D)}{2}}(1+r+u)^{\frac{(D-2)}{2}}}\right].
		\end{align}
		Since $1+r(u)+u\geq \frac{1}{2}\kappa(1+r_1+u_1)$, we obtain
		\begin{align}
		P_{13}\leq\frac{C(\rho(x)+x^6)y}{(1+u)^{3(2k-D)}(1+r+u)^\frac{3D}{2}},
		\end{align}
		where $\rho(x)=Cx^5(d+x^3+x^5+x^p+x^{p+2})\exp\left[C(x^2+x^4+x^{p+1})\right]$.
		
		\item Estimate for $P_{14}$\\
		Using (\ref{estimasi htilde}) and (\ref{Q1-Q2}) we obtain
		\begin{align}
		P_{14}=\frac{|Q_1-Q_2|}{2r}|g_1||\tilde{h}_1|\leq\frac{Cx^2y}{(1+u)^{\frac{3(2k-D)}{2}}(1+r+u)^{\frac{D}{2}}}.
		\end{align}
		\item Estimate for $P_{15}$\\
		In the view of (\ref{estimasi htilde}), (\ref{estimate Q}), and (\ref{g1-g2}) we get
		\begin{align}
		P_{15}=\frac{|Q_2|}{2r}|g_1-g_2||\tilde{h}_1|
		\leq\frac{Cx^4y}{(1+u)^{\frac{5(2k-D)}{2}}(1+r+u)^{\frac{(3D-4)}{4}}}.
		\end{align}
		\item Estimate for $P_{16}$\\
		From (\ref{estimate Q}) and (\ref{h1-h2}) we obtain
		\begin{align}
		P_{16}=\frac{|Q_2||g_2|}{2r}|\tilde{h}_1-\tilde{h}_2|
		\leq\frac{Cx^2y}{(1+u)^{\frac{3(2k-D)}{2}}(1+r+u)^{\frac{D}{2}}}.
		\end{align}
		\item Estimate for $P_{17}$\\
		We firstly write the following estimate
		\begin{align}
		|{A_0}_1-{A_0}_2|\leq&\int_{0}^{r}\frac{|Q_1-Q_2|}{s^{D-2}}g_1\mathrm{d}s+\int_{0}^{r}\frac{|g_1-g_2||Q_2|}{s^{D-2}}\mathrm{d}s\nonumber\\
		\leq&\frac{C x y}{(1+u)^{2k-D}}\int_{0}^{r}\frac{\mathrm{d}s}{(1+s+u)^{D-2}}+\frac{C x^3y}{(1+u)^{2(2k-D)}}\int_{0}^{r}\frac{\mathrm{d}s}{(1+s+u)^{2(D-2)}}\nonumber\\
		\leq& \frac{C(x+x^3)yr^{D-3}} {(1+u)^{2k-3}(1+r+u)^{D-3}}.
		\end{align}
		Thus,
		\begin{align}
		P_{17}=&|\mathcal{F}_2||{A_0}_1-{A_0}_2|\nonumber\\
		\leq&\frac{C(d+x^3+x^5+x^p+x^{p+2})\exp\left[C(x^2+x^4+x^{p+1})\right]}{\kappa^{k-1}(1+r_1+u_1)^{k-1}}\frac{C(x+x^3)yr^{D-3}} {(1+u)^{2k-3}(1+r+u)^{D-3}}.
		\end{align}
		Since $1+r(u)+u\geq \frac{1}{2}\kappa(1+r_1+u_1)$, we obtain
		\begin{align}
		P_{17}\leq\frac{C\omega(x)y}{(1+u)^{2k-3}(1+r+u)^{k-1}},
		\end{align}
		where
		\begin{align}
		\omega(x)=C(x+x^3)(d+x^3+x^5+x^p+x^{p+2})\exp\left[C(x^2+x^4+x^{p+1})\right].
		\end{align}
	\end{enumerate}
	
	Finally, we represent the estimate for (\ref{phi tilde}) as follows
	\begin{align}\label{varphi}
	|\tilde{\varphi}|\leq\frac{C(\alpha(x)+\beta(x)+\gamma(x)+\sigma(x)+\eta(x)+\rho(x)+\omega(x)+\lambda(x))y}{(1+u)^{\frac{(2k-D)k}{2}}(1+r+u)^{\frac{D}{2}}},
	\end{align}
	where $\alpha(x), \beta(x), \gamma(x), \sigma(x), \eta(x), \rho(x), \omega(x)$ are defined in (\ref{alpha}), (\ref{beta}), (\ref{gamma}), (\ref{sigma}), (\ref{eta}), (\ref{rho}), (\ref{omega}), and (\ref{lambda}) respectively.
	
	\section*{Acknowledgments}
	The work of this research is supported by  ITB Research Grant 2024.
	
	\section*{Data Availability}
	This manuscript has no associated data.
	
	\section*{Conflict of Interest}
	No conflict of interest in this paper.
	
	\appendix
	
}
\end{document}